\documentclass[11pt,onecolumn]{article}
\usepackage{authblk}
\usepackage{microtype}
\usepackage{xspace}
\usepackage{picins,wrapfig}

\usepackage{times}
\usepackage{amsmath}
\usepackage{amssymb}
\usepackage{multirow}
\usepackage{xspace}
\usepackage{theorem}
\usepackage{graphicx}   
\usepackage{ifpdf}
\usepackage[colorlinks, citecolor=black]{hyperref}
\usepackage{latexsym}
\usepackage{euscript}
\usepackage{xspace}
\usepackage{color}
\usepackage{makeidx}
\usepackage{picins,wrapfig}
\usepackage{xypic}
\xyoption{all}

\usepackage{amsmath}
\usepackage{amssymb}
\usepackage{algorithm}
\usepackage{algorithmic}
\usepackage{theorem}
\usepackage{graphicx}
\usepackage{ifpdf}
\usepackage{xspace}
\usepackage{xypic}
\usepackage{picins,wrapfig}
\usepackage{color}
\xyoption{curve}

\theoremstyle{plain}
\newtheorem{theorem}{Theorem}
\newtheorem{lemma}[theorem]{Lemma}
\newtheorem{claim}[theorem]{Claim}
\newtheorem{cor}[theorem]{Corollary}
\newtheorem{obs}[theorem]{Observation}
\newtheorem{proposition}[theorem]{Proposition}

\newtheorem{definition}[theorem]{Definition}
\newenvironment{proof}{{\em Proof:}}{\hfill{\hfill\rule{2mm}{2mm}}}

\definecolor{darkred}{rgb}{1, 0.1, 0.3}

\newcommand{\etal}      {et al.\@\xspace}

\newcommand {\mm}[1] {\ifmmode{#1}\else{\mbox{\(#1\)}}\fi}
\newcommand{\denselist}{\itemsep 0pt\parsep=1pt\partopsep 0pt}
\newcommand{\eps}{{\varepsilon}}
\newcommand{\reals}	{{\rm I\!\hspace{-0.025em} R}}

\newcommand{\graphone} {\mathrm{G}_1}
\newcommand{\graphtwo} {\mathrm{G}_2}
\newcommand{\perspace}	 {\mathbb{D}}
\newcommand{\done}[1]	{{d_{G_1, {#1}}}}
\newcommand{\dtwo}[1]	{{d_{G_2, {#1}}}}
\newcommand{\perone}[1] {{\mathrm{P}_{#1}}}
\newcommand{\pertwo}[1] {{\mathrm{Q}_{#1}}}
\newcommand{\bp}		{\mathbf{s}}
\newcommand{\nbp}		{\mathbf{t}}

\newcommand{\setone}	{{\mathcal{C}}}
\newcommand{\settwo}	{{\mathcal{F}}}
\newcommand{\spdist}	{persistence-distortion\xspace}
\newcommand{\Spdist}	{Persistence-distortion\xspace}
\newcommand{\spd}		{PD}
\newcommand{\dsp}		{\mathrm{d_{PD}}}
\newcommand{\disdsp}	{\mathrm{\widehat{d}_{PD}}}
\newcommand{\dgh}		{\mathrm{d_{GH}}}
\newcommand{\dfd}		{\mathrm{d_{FD}}}
\newcommand{\matching}	{{\mathcal{M}}}
\newcommand{\cmatch}		{{\mathcal{M}}_c}
\newcommand{\height}		{\mathrm{height}}
\newcommand{\leftmap}		{\phi_{\to}}
\newcommand{\rightmap}		{\phi_{\leftarrow}}

\newcommand{\Vone}		{V_1}
\newcommand{\Vtwo}		{V_2}
\newcommand{\Eone}		{E_1}
\newcommand{\Etwo}		{E_2}

\newcommand{\birthp}		{u_\perb}
\newcommand{\deathp}		{u_\perd}
\newcommand{\perb}			{\mathsf{b}}
\newcommand{\perd}			{\mathsf{d}} 
\newcommand{\geod}		{\mathsf{g}}
\newcommand{\geode}		{\geod}
\newcommand{\deathe}		{e_\perb}
\newcommand{\ebase}		{\sigma}

\newcommand{\recR}			{\Omega} 
\newcommand{\pertube}		{\Pi}
\newcommand{\traj}			{\pi}

\newcommand{\onePD}		{F}
\newcommand{\length}		{\mathrm{Len}}
\newcommand{\birthdeath}	{birth-death}
\newcommand{\Drb}			{D}
\newcommand{\pcritical}		{X}
\newcommand{\dcomp}		{\Lambda}
\newcommand{\augdcomp}	{\widehat{\dcomp}} 
\newcommand{\Lenv}		{{\mathcal{L}}}
\newcommand{\HH}			{{\mathrm{H}}}

\begin{document}

\title{Comparing Graphs via Persistence Distortion}

\author{Tamal Dey\thanks{Department of Computer Science and Engineering, The Ohio State University, Columbus, OH, USA. Emails: \texttt{tamaldey, shiday, yusu@cse.ohio-state.edu}},~~~ Dayu Shi$^*$,~~ Yusu Wang$^*$}

\date{}

\maketitle

\begin{abstract}
Metric graphs are ubiquitous in science and engineering. For example, many data are drawn from hidden spaces that are graph-like, such as the cosmic web. A metric graph offers one of the simplest yet still meaningful ways to represent the non-linear structure hidden behind the data. In this paper, we propose a new distance between two finite metric graphs, called the \spdist{} distance, which draws upon a topological idea. This topological perspective along with the metric space viewpoint provide a new angle to the graph matching problem. Our \spdist{} distance has two properties not shared by previous methods: First, it is stable against the perturbations of the input graph metrics. Second, it is a \emph{continuous} distance measure, in the sense that it is defined on an alignment of the underlying spaces of input graphs, instead of merely their nodes. This makes our \spdist{} distance robust against, for example, different discretizations of the same underlying graph.

Despite considering the input graphs as continuous spaces, that is, taking all points into account, we show that we can compute the \spdist{} distance in polynomial time. The time complexity for the discrete case where only graph nodes are considered is much faster.
We also provide some preliminary experimental results to demonstrate the use of the new distance measure. 

 \end{abstract}

\section{Introduction}

Many data in science and engineering are drawn from hidden spaces which are 
graph-like, such as the cosmic web \cite{SPK11} and road 
networks \cite{ACC12,CS14}. Furthermore, as modern data become 
increasingly complex, understanding them with a 
simple yet still meaningful structure becomes important. Metric graphs 
equipped with a metric derived from the data can provide such a simple 
structure~\cite{GSBW11,OE11}. 
They are graphs where each edge is associated with a length inducing
the metric of shortest path distance.
The comparison of the representative metric graphs can benefit 
classification of data, a fundamental task in processing them.
This motivates the study of metric graphs in the context of 
matching or comparison.

To compare two objects, one needs a notion of distance in the
space where the objects are coming from. Various distance measures 
for graphs 
have been proposed in the literature with
associated matching algorithms.
We approach this problem with two new perspectives: (i) We aim to develop a 
distance measure which is both meaningful and stable
against metric perturbations, and at the same time amenable to polynomial
time computations.
(ii) Unlike most previous distance measures which are \emph{discrete} in the sense that only graph node alignments are considered, we aim 
for a distance measure that is \emph{continuous}, that is, alignment for
all points in the underlying space of the metric graphs are
considered. 

\paragraph{Related work.}
To date, the large number of proposed graph matching algorithms fall into two broad categories: exact graph matching methods and inexact graph matching (distances between graphs) methods. 
The exact graph matching, also called the graph isomorphism problem, checks whether there is a bijection between the node sets of two input graphs that also induces a bijection in their edge sets. While polynomial time algorithms exist for many special cases, e.g., \cite{AHU74,HW74,L82},
for general graphs, it is not known whether there exists polynomial time algorithm for the graph isomorphism problem, 
(despite the ground-breaking recent work by Babai showing that it can be solved in quasi-polynomial time \cite{Babai16}). 
Nevertheless, given the importance of this problem, there are various exact graph matching algorithms developed in practice. 
Usually, these methods employ some pruning techniques aiming to  reduce the search space for identifying graph isomorphisms. See \cite{FSV01} for comparisons of various  graph isomorphism testing methods. 

In real world applications, input graphs often suffer from noise and deformation, and it is highly desirable to obtain a \emph{distance} between two input graphs beyond the binary decision of whether they are the same (isomorphic) or not. This is referred to as inexact graph matching in the field of pattern recognition, and various distance measures have been proposed. One line of work is based on graph edit distance which is NP-hard to compute \cite{ZTWFZ09}. Many heuristic methods, using for example $A^*$ algorithms, have been proposed to address the issue of high computational complexity, see the survey \cite{GXTL10} and references within.
One of the main challenges in comparing two graphs is to determine how
"good" a given alignment of graph nodes is in terms of the quality of 
the pairwise relations between those nodes. Hence matching two graphs naturally leads to an integer quadratic programming problem (IQP), which is a NP-hard problem. Several heuristic methods have been proposed to approach this optimization problem, such as the annealing approach of \cite{GR96}, iterative methods of \cite{LHS09,WW04} and probabilistic approach in \cite{ZS08}. Finally, there have been several methods that formulate the optimization problem based on spectral properties of graphs. For example, in \cite{U98}, the author uses the eigendecomposition of adjacency matrices of the input graphs to derive an expression of an orthogonal matrix which optimizes the objective function. In \cite{CSS07,LH05}, the principal eigenvector of a ``compatibility'' matrix of the input graphs is used to obtain correspondences between input graph nodes. Recently in \cite{HRG13}, Hu et. al proposed the general and descriptive \emph{Laplacian family signatures} to build the compatibility matrix and model the graph matching problem as an integer quadratic program.

\paragraph{New work.}
Unlike previous approaches, we view input graphs as \emph{continuous} metric spaces. Intuitively, we assume that our input is a finite graph $G = (V, E)$ where each edge is assigned a positive length value. We now consider $G$ as a metric space $(|G|, d_G)$ on the underlying space $|G|$ of $G$, with metric $d_G$ being the shortest path metric in $|G|$. 
Given two metric graphs $G_1$ and $G_2$, a natural way to define their distance is to use the so-called Gromov-Hausdorff distance \cite{Gromov99,M07} that measures the metric distortion between these two metric spaces. Unfortunately, it is NP-hard to even approximate the Gromov-Hausdorff distance for graphs within a constant factor \cite{AFNSW15}. 
Instead, we propose a new metric, called the \emph{\spdist{} distance} $\dsp(G_1,G_2)$, which draws upon a topological idea and is computable in polynomial time with techniques from computational geometry. This provides a new angle to the graph comparison problem. The distance that we define has several nice properties: 

\parpic[r]{\includegraphics[height=2cm]{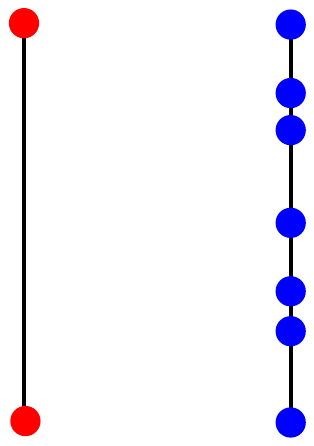}}
(1) The \spdist{} distance takes into account all points in the geometric realization of the input graphs, while all previous graph matching algorithms align only graph nodes. Hence our \spdist{} distance is insensitive to different discretization of the same graph: For example, the two geometric graphs on the right are equivalent as metric graphs, and thus the \spdist{} distance between them is zero.  

\begin{wrapfigure}{r}{0.35\textwidth}
\begin{tabular}{cc}
\includegraphics[height=1.4cm]{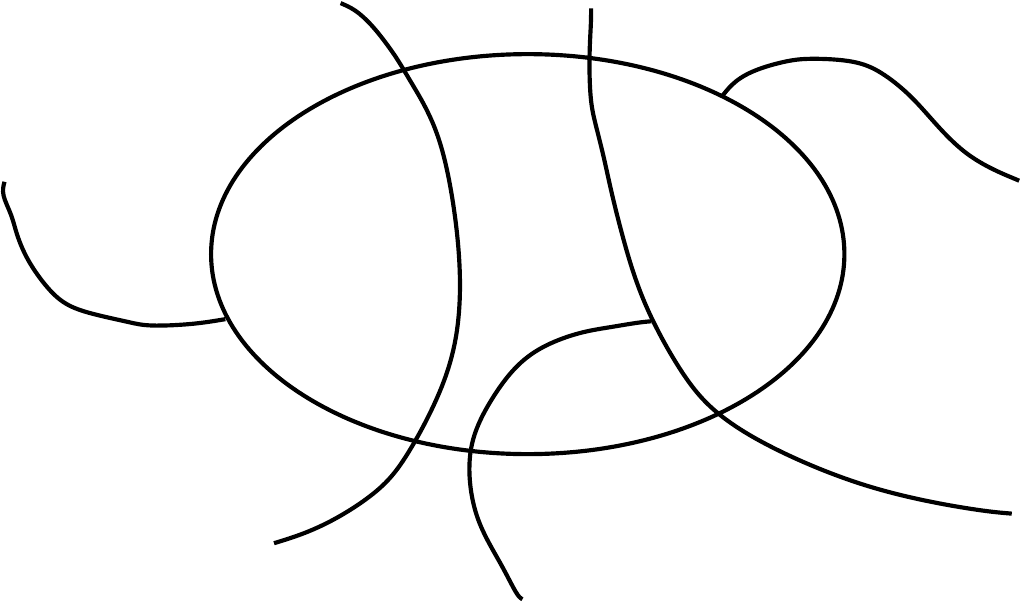} & \includegraphics[height=1.4cm]{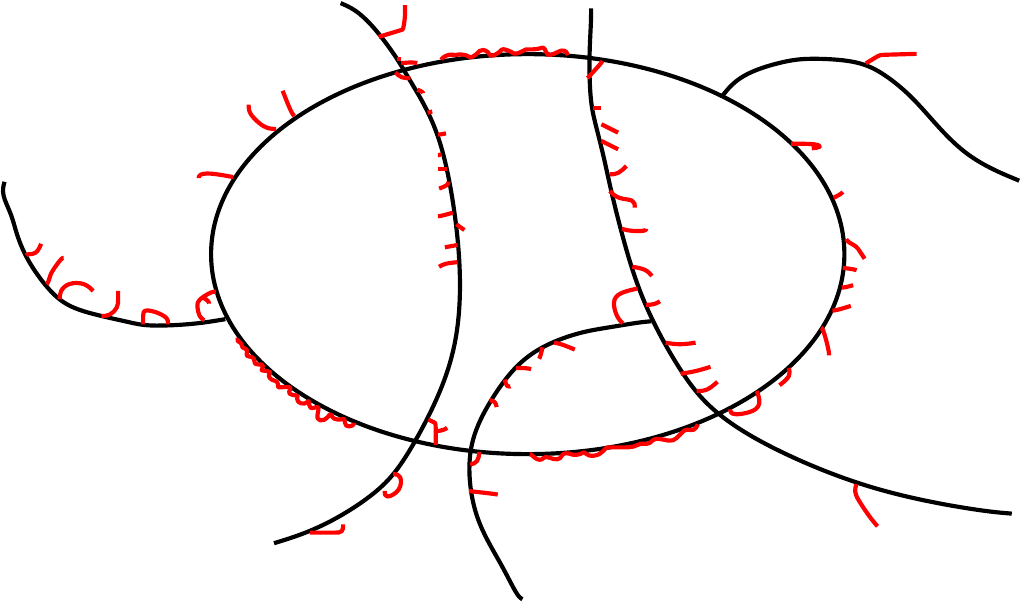}\\
$G_1$ & $G_2$
\end{tabular}
\end{wrapfigure}
(2) In Section \ref{sec:stability}, we show that the \spdist{} distance $\dsp(G_1,G_2)$ is stable w.r.t. changes to input metric graphs as measured by the Gromov-Hausdorff distance. 
For example, the two geometric graphs on the right have small \spdist{} distance. (Imagine that they are the reconstructed road networks from noisy data sampled from the same road systems.)  

(3) Despite that the \spdist{} distance is a \emph{continuous} measure which considers all points in a geometric realization of the input graphs, we show in Section \ref{sec:compcontinuous} that it can be computed in polynomial time ($O(m^{12}\log m)$ where $m$ is the total number of nodes and edges of the input graphs). We note that the \emph{discrete} version of the \spdist{} distance, where only graph nodes are considered (much like in previous graph matching algorithms), can be computed much more efficiently in $O(n^2 m^{1.5}\log m)$ time, where $n$ is the number of graph nodes in input graphs. 



Finally, we also provide some preliminary experimental results to demonstrate the use of the \spdist{} distance. 

\section{Notations and Proposed Distance Measure for Graphs}
\label{sec:notations}

\paragraph{Metric graphs.} 
A metric graph is a metric space $(M,d)$ where $M$ is the
underlying space of a finite $1$-dimensional simplicial complex.
Given a graph $G=(V,E)$ and a weight function 
$\length: E\rightarrow\mathbb{R}^+$ on its edge set $E$ (assigning length to edges in $E$), we can
associate a metric graph $(|G|,d_G)$ to it as follows.
The space $|G|$ is a geometric realization of $G$. Let $|e|$
denote the image of an edge $e\in E$ in $|G|$. To define the
metric $d_G$, we consider the arclength parameterization 
$e:[0,\length(e)]\rightarrow |e|$ for every edge $e\in E$ and
define the distance between any two points $x, y \in |e|$ as 
$d_G(x,y) = |e^{-1}(y) - e^{-1}(x)|$. This in turn provides the length of a path 
$\pi(z, w)$ between two points $z, w\in |G|$ that are not necessarily 
on the same edge in $|G|$, by simply summing up the lengths of the 
restrictions of this path to edges in $G$. 
Finally, given any two points $z, w\in |G|$, the distance $d_G(z, w)$ is 
given by the minimum length of any path connecting $z$ and $w$ in $|G|$. 

In what follows, we do not distinguish between $|\cdot|$ and
its argument and write $(G, d_G)$ to denote the metric
graph $(|G|,d_G)$ for simplicity.
Furthermore, for simplicity in presentation, we abuse the notations slightly and refer to the metric graph as $G=(V,E)$, with the understanding that 
$(V, E)$ refers to the graph representing the metric space $(G, d_G)$.  
Finally, we refer to any point $x\in G$ (i.e, $x\in |G|$) as a point, while a point $x\in V$ as a \emph{graph node}. 

\paragraph{Background on persistent homology.}
The definition of our proposed distance measure for two metric graphs relies on the so-called persistence diagram induced by a scalar function. 
We refer the readers to resources such as \cite{EH09,ELZ02} for 
formal discussions on persistent homology and related developments. 
Below we only provide an intuitive and informal description of the persistent homology induced by a function under our simple setting. 

Let $f: X \to \reals$ be a continuous real-valued function defined on a topological space $X$. We want to understand the structure of $X$ from the perspective of the scalar function $f$: Specifically, let $X^{\alpha} := \{ x\in X \mid f(x) \ge \alpha\}$ denote the \emph{super-level set}\footnote{In the standard formulation of persistent homology of a scalar field, the \emph{sub-level set} $X_\alpha = \{ x\in X \mid f(x) \le \alpha\}$ is often used. We use super-level sets which suit the specific functions that we use. } of $X$ w.r.t. $\alpha \in \reals$. 
Now as we sweep $X$ top-down by decreasing the value $\alpha$; the sequence of super-level sets equipped with natural inclusion maps gives rise to a  \emph{filtration of $X$ induced by $f$}: 
\begin{equation} \label{eqn:filtration}
X^{\alpha_1} \subseteq X^{\alpha_2} \subseteq \cdots \subseteq X^{\alpha_m} = X, ~~~~\text{for}~~ \alpha_1 > \alpha_2 > \cdots > \alpha_m. 
\end{equation}
We track how the topological features captured by the so-called homology 
classes of the super-level sets change. 
In particular, as $\alpha$ decreases, sometimes new topological features are ``born'' at time $\alpha$, that is, new families of homology classes are created in $\HH_k(X^\alpha)$, the $k$-th homology group of $X^\alpha$. Sometimes, existing topological features disappear, i.e, some homology classes become trivial in $\HH_k(X^\beta)$ for some $\beta < \alpha$. The \emph{persistent homology} captures such \emph{birth} and \emph{death} events, and summarizes them in the \emph{persistence diagram} $\mathrm{Dg}_k (f)$. Specifically, $\mathrm{Dg}_k(f)$ consists of a set of points $\{ (\alpha, \beta) \in \reals^2 \}$ in the plane, where each $(\alpha,\beta)$ indicates a homological feature created at time $\alpha$ and killed entering time $\beta$. 


In our setting, the domain $X$ will be the underlying space of a metric graph $G$. The specific function that we use later is the geodesic distance to a fixed basepoint $\bp\in G$, that is, we consider $f: G \to \reals$ where $f(x) = d_G(\bp, x)$ for any $x \in G$. 
We are only interested in the 0th-dimensional persistent homology ($k =0$ in the above description), which simply tracks the connected components in the super-level set as we decrease $\alpha$. 

\begin{figure}[htp]
\begin{center}
\begin{tabular}{ccccccc}
\includegraphics[height=2.8cm]{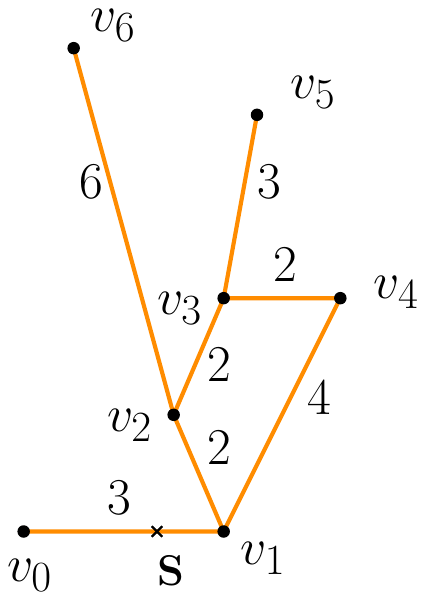} &\hspace*{0.05in} &\includegraphics[height=2.8cm]{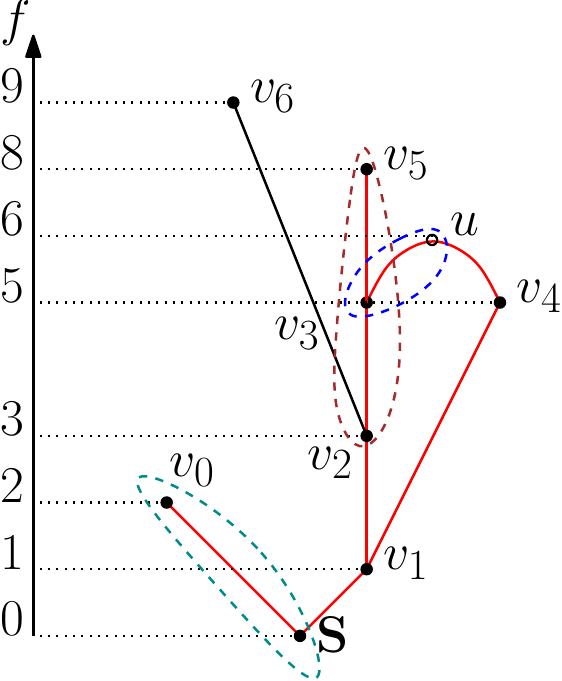} & \hspace*{0.0in}& \includegraphics[height=2.5cm]{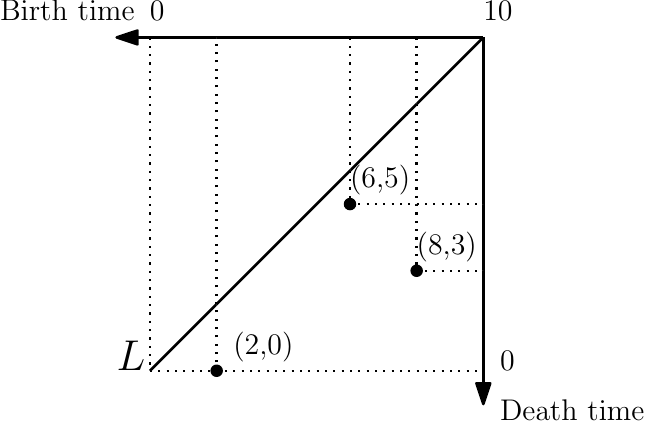} & & \includegraphics[height=2.5cm]{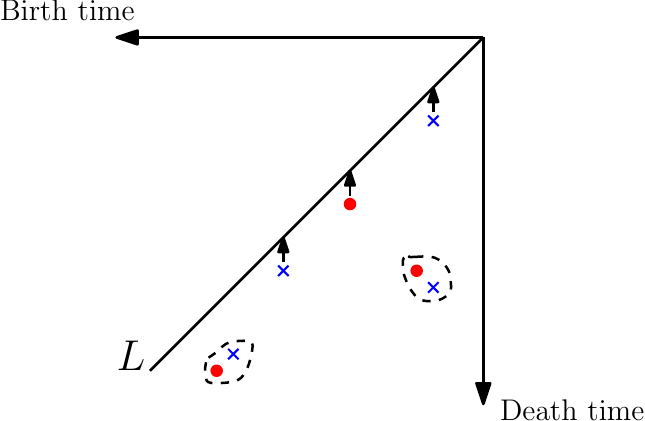}\\
(a) & & (b) & & (c) & & (d)
\end{tabular}
\end{center}
\vspace*{-0.15in}\caption{(a) A graph with basepoint $\bp$: edge length is marked for each edge. (b) The function $f = d_G(\bp, \cdot)$. We also indicate critical-pairs. (c) Persistence diagram $\mathrm{Dg}_0 f$: E.g, the persistence-point $(6,5)$ is generated by critical-pair $(u, v_3)$. (d) A partial matching between the red points and blue points (representing two persistence diagrams). Some points are matched to the diagonal $L$. 
\label{fig:graphexample}}
\end{figure}
Figure \ref{fig:graphexample} gives an example of the 0-th persistence diagram $\mathrm{Dg}_0 (f)$ with the basepoint $\bp$ in edge $(v_0,v_1)$. 
As we sweep the graph top-down in terms of the geodesic function $f$, a new connected component is created as we pass through a \emph{local maximum} $\birthp$ of the function $f = d_G(\bp, \cdot)$. 
A local maximum of $f$ such as $u$ in
Figure \ref{fig:graphexample} (b)
is not necessarily a graph node from $V$.
Two connected components in the super-level set can only merge at an \emph{up-fork saddle} $\deathp$ of the function $f$: The up-fork saddle $\deathp$ is a point 
that has a neighborhood with at least two branches incident on $\deathp$ whose function values are larger than $\deathp$.
Each point $(\perb,\perd)$ in the persistence diagram is called a \emph{persistence point}, 
corresponding to the creation and death of some connected component: At time $\perb$, a new component is created in $X^\perb$ at a local maximum $\birthp \in G$ with $f(\birthp) = \perb$. 
At time $\perd$ and
at an up-fork saddle $\deathp \in G$ with $f(\deathp) = \perd$,
this component merges with another component created 
earlier. 
We refer to the pair of points $(\birthp, \deathp)$ from the graph $G$ as the \emph{critical-pair} corresponding to the persistent point $(\perb, \perd)$. We call $\perb$ and $\perd$ the \emph{birth-time} and \emph{death-time}, respectively. 
The plane containing the persistence diagram is called the \emph{\birthdeath{} plane}.
%

Finally, given two finite persistence diagrams $\mathrm{Dg} = \{ p_1, \ldots, p_\ell \in \reals^2 \}$ and $\mathrm{Dg}' = \{ q_1, \ldots, q_k \in \reals^2 \}$, a common distance measure for them, the \emph{bottleneck distance} $d_B (\mathrm{Dg}, \mathrm{Dg}')$ \cite{CEH07}, is defined as follows: 
Consider $\mathrm{Dg}$ and $\mathrm{Dg}'$ as two finite sets of points in the plane (where points may overlap). 
Call $L = \{(x,x) \in \reals^2\}$ the \emph{diagonal} of the birth-death plane. 

\begin{definition}\label{def:bottleneck}
A \emph{partial matching} $C$ of $\mathrm{Dg}$ and $\mathrm{Dg}'$ is a relation $C : (\mathrm{Dg}\cup L) \times (\mathrm{Dg}' \cup L)$ such that each point in $\mathrm{Dg}$ is either matched to a unique point in $\mathrm{Dg}'$, or mapped to its closest point (under $L_\infty$-norm) in the diagonal $L$; and the same holds for points in $\mathrm{Dg}'$. See Figure \ref{fig:graphexample} (d). The bottleneck distance is defined as $d_B(\mathrm{Dg}, \mathrm{Dg}') = \min_{C} \max_{(p,q)\in C} \|p - q\|_\infty$, where $C$ ranges over all possible partial matchings of $\mathrm{Dg}$ and $\mathrm{Dg}'$. 
We call the partial matching that achieves the bottleneck distance $d_B(\mathrm{Dg}, \mathrm{Dg}')$ as the \emph{bottleneck matching}. 
\end{definition}

\paragraph{Proposed \spdist{} distance for metric graphs.}

Suppose we are given two metric graphs $(\graphone, d_{G_1})$ and $(\graphtwo, d_{G_2})$. Let $(V_1, E_1)$ and $(V_2, E_2)$ denote the node set and edge set for $\graphone$ and $\graphtwo$, respectively. 
Set $n = \max \{ |V_1|, |V_2| \}$ and $m = \max \{|E_1|, |E_2| \}$. 

Choose any point $\bp \in \graphone$ as the base point, and consider the shortest path distance function $\done{\bp}: \graphone \to \reals$ defined as $\done{\bp}(x) = d_{G_1}(\bp, x)$ for any point $x\in \graphone$. 
Let $\perone{\bp}$ denote the 0-th dimensional persistence diagram $\mathrm{Dg}_0 (\done{\bp})$ induced by the function $\done{\bp}$. 
Define $\dtwo{\nbp}$ and $\pertwo{\nbp}$ similarly for any base point $\nbp \in \graphtwo$ for the graph $\graphtwo$.  
We map the graph $\graphone$ to the set of (infinite number of) points in the space of persistence diagrams $\perspace$, given by $\setone:= \{ \perone{\bp} \mid \bp \in \graphone \}$. Similarly, map the graph $\graphtwo$ to $\settwo:= \{ \pertwo{\nbp} \mid \nbp \in \graphtwo \}$. 

\begin{definition}\label{def:spdist}
The \emph{\spdist{} distance between $\graphone$ and $\graphtwo$}, denoted by $\dsp(\graphone,\graphtwo)$, is the Hausdorff distance $d_H(\setone, \settwo)$ between the two sets $\setone$ and $\settwo$ where the distance between two persistence diagrams is measured by the bottleneck distance. In other words, 
$$\dsp(\graphone,\graphtwo) = d_H(\setone, \settwo) = \max \{ ~\max_{\mathrm{P} \in \setone} \min_{\mathrm{Q} \in \settwo} d_B(\mathrm{P}, \mathrm{Q}), ~ \max_{\mathrm{Q} \in \settwo} \min_{\mathrm{P} \in \setone} d_B(\mathrm{P}, \mathrm{Q}) ~\}.
$$

\end{definition}

\paragraph{Remark.} (1) We note that if two graphs are isomorphic, then $\dsp(\graphone,\graphtwo)=0$. The inverse unfortunately is not true. (See Figure \ref{fig:failurecase} for an example where two graphs have $\dsp(\graphone,\graphtwo)=0$, but they are not isomorphic.)
Hence $\dsp$ is a pseudo-metric (it inherits the triangle-inequality property from the Hausdorff distance). 
(2) While the above definition uses only the 0-th persistence diagram for the geodesic distance functions, all our results hold with the same time complexity when we also include the \emph{1st-extended persistence diagram} \cite{CEH09} 
or equivalently \emph{1st-interval persistence diagram} \cite{DW07}
for each geodesic distance function $\done \bp$ (resp. $\dtwo \nbp$). 

\begin{figure}[htbp]
\begin{center}
\begin{tabular}{ccc}
\includegraphics[height=2cm]{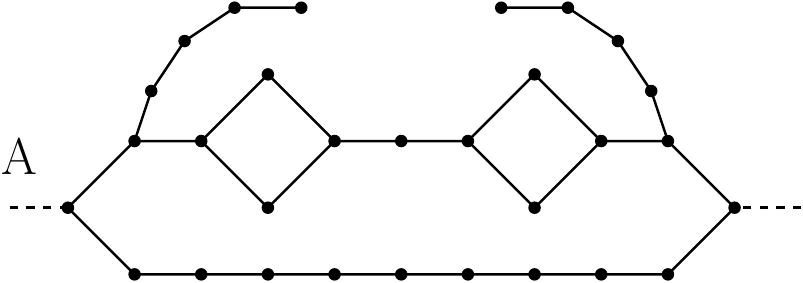} & \includegraphics[height=2cm]{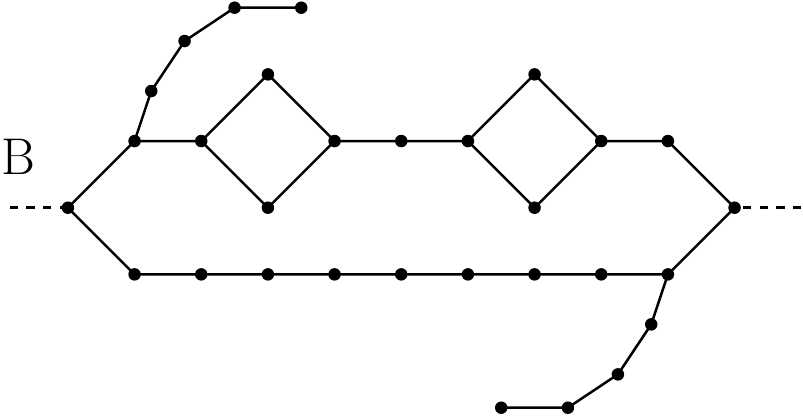} & \includegraphics[height=2cm]{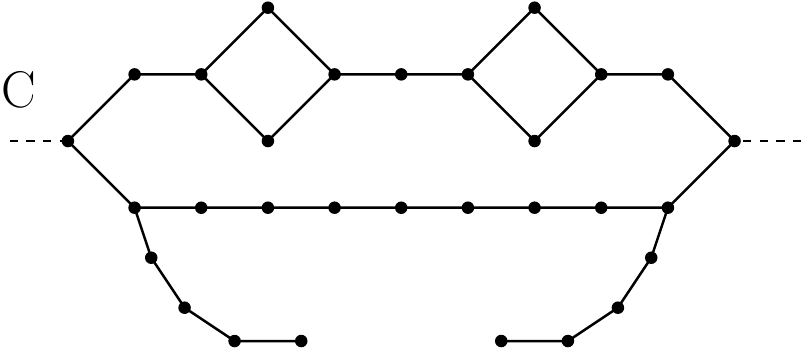} \\
A & B & C
\end{tabular}
\begin{tabular}{ccc}
\includegraphics[height=2.6cm]{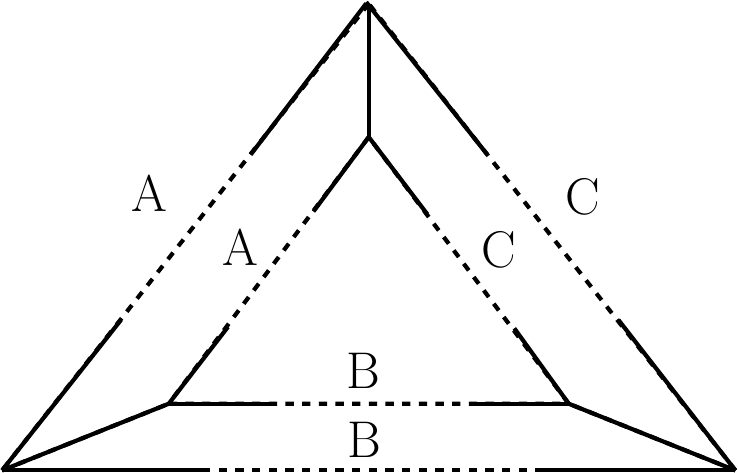} & \hspace*{0.2in} & \includegraphics[height=2.6cm]{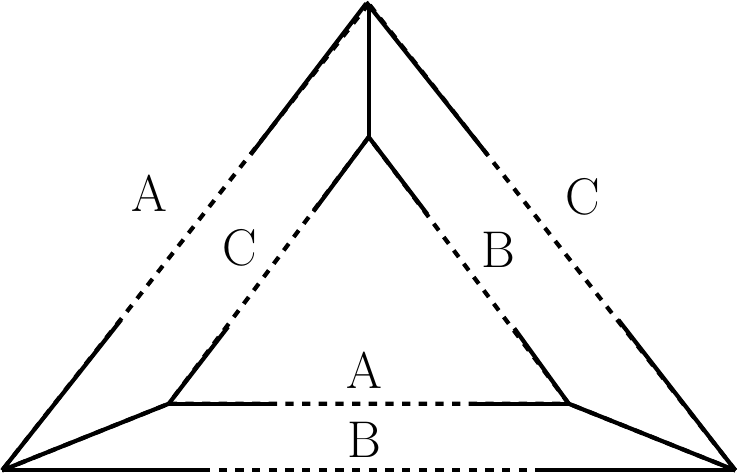}\\
$G_1$ & & $G_2$ 
\end{tabular} 
\end{center}
\vspace*{-0.1in}\caption{\label{fig:failurecase} In the top row, we show three components $A, B$ and $C$, that will be used to construct the two input metric graphs $G_1$ and $G_2$ in the bottom row. All edges are of length 1. Each of these components has the property that, from a basepoint outside these components, the persistence diagram w.r.t. the geodesic distance remains the same (and hence cannot be distinguished). Graphs $G_1$ and $G_2$ are not isomorphic. However, they are mapped to the same image set in the space of persistence diagrams, and hence the \spdist{} distance between them is zero; i.e, $\dsp(\graphone,\graphtwo) = 0$ . 
}
\end{figure}

\section{Stability of \spdist distance}
\label{sec:stability}

\paragraph{Gromov-Hausdorff distance.}
There is a  natural way to measure metric distortion between metric spaces (thus for metric graphs) by the Gromov-Hausdorff distance \cite{Gromov99,BBI01}. 
Given two metric spaces $\mathcal{X}=(X, d_X)$ and $\mathcal{Y} = (Y, d_Y)$, a \emph{correspondance} between $\mathcal{X}$ and $\mathcal{Y}$ is a relation $\matching: X \times Y$ such that 
(i) for any $x\in X$, there exists $(x, y) \in \matching$ and (ii) for any $y' \in Y$, there exists $(x', y') \in \matching$. 
The \emph{Gromov-Hausdorff} distance between $\mathcal{X}$ and $\mathcal{Y}$ is
\begin{equation}
\dgh(\mathcal{X}, \mathcal Y) = \frac{1}{2} \inf_{\matching} \max_{(x_1,y_1), (x_2, y_2) \in \matching} |d_{X}(x_1, x_2) - d_{Y}(y_1,y_2)| ,  
\label{eqn:dgh}
\end{equation}
where $\matching$ ranges over all correspondences of $X \times Y$. 
The Gromov-Hausdorff distance is a natural distance between two metric spaces; see \cite{M07} for more discussions. 
Unfortunately, so far, there is no efficient (polynomial-time) algorithm to compute or approximate this distance, even for special metric spaces -- In fact, it has been recently shown that even the \emph{discrete} Gromov-Hausdorff distance for metric trees (where only tree nodes are considered) is NP-hard to compute as well as to approximate within a constant factor (see footnote 1). 
In contrast, as we show in Section \ref{sec:compdisc} and \ref{sec:compcontinuous}, the \spdist{} distance can be computed in polynomial time. 

On the other hand, we have the following stability result, 
which intuitively suggests that the \spdist{} distance is a weaker relaxation of the Gromov-Hausdorff distance. 
The proof of this theorem leverages a recent result on measuring distances between the Reeb graphs \cite{BGW14}. 
\begin{theorem}[Stability]
$\dsp(\graphone,\graphtwo) \le 6\dgh(\graphone, \graphtwo)$. 

By triangle inequality, this also implies that given two metric graphs $\graphone$ and $\graphtwo$ and their perturbations $\graphone'$ and $\graphtwo'$, respectively, we have that:
$$\dsp(\graphone', \graphtwo') \le \dsp(\graphone, \graphtwo) + 6\dgh(\graphone, \graphone') + 6 \dgh(\graphtwo, \graphtwo'). $$
\label{thm:stability}
\end{theorem}
\paragraph{Proof of Theorem \ref{thm:stability}. }  
The remainder of this section devotes to the proof of the above theorem. 

Given two input metric graphs $(\graphone, d_{G_1})$ and $(\graphtwo, d_{G_2})$, 
set $\delta = \dgh(\graphone, \graphtwo)$ to be the Gromov-Hausdorff distance between $G_1$ and $G_2$. 
Assume that the correspondence $\matching^*$ achieves this metric distortion distance $\delta$ \footnote{It is possible that $\delta$ is achieved in the limit, in which case, we consider a sequence of $\eps$-correspondences whose corresponding metric distortion distance converges to $\delta$ as $\eps$ tends to $0$. For simplicity, we assume that $\delta$ can be achieved by the correspondence $\matching^*$. }. 
Now for any point $\bp \in \graphone$, there must exist $\nbp \in \graphtwo$ such that $(\bp, \nbp) \in \matching^*$. 
We will now show that $d_B(\perone{\bp}, \pertwo \nbp) \le 6\delta$. 
Symmetrically, we show that for any $\nbp \in \graphtwo$, there is $(\bp, \nbp) \in \matching^*$ such that $d_B(\perone{\bp}, \pertwo \nbp) \le 6\delta$. 
Since such a $\nbp$ can be found for any point $\bp \in \graphone$, and symmetrically, such an $\bp$ can be found for any $\nbp \in \graphtwo$, it then follows that the Hausdorff distance between $\setone$ and $\settwo$ is bounded from above by $6\delta$, proving the theorem. 

We will prove $d_B(\perone{\bp}, \pertwo \nbp) \le 6\delta$ for $(\bp,\nbp) \in \matching^*$ with the help of another distance, the so-called functional-distortion distance between two Reeb graphs introduced in \cite{BGW14}. We recall its definition below. 

The functional-distortion distance is defined between two graphs $\graphone$ and $\graphtwo$, with a function $f: \graphone \to \reals$ and $g: \graphtwo \to \reals$ defined on each of them, respectively. 
(In our case, $f$ will later be taken as the shortest path distance function $f = \done\bp$ and $g$ will be taken as $g = \dtwo \nbp$.)  
First, we define the following \emph{(pseudo-)metric on the input graphs as induced by $f$ and $g$}, respectively. (It is important to note that these metrics are different from the path-length distance metrics $d_{G_1}$ and $d_{G_2}$ that input graphs already come with.) 
Specifically, given two points $x_1, x_2 \in \graphone$, define 
\begin{align}
d_f(x_1, x_2) &= \min_{\pi: x_1 \leadsto x_2} \height(\pi), \label{eqn:df}
\end{align}
where  $\pi$ ranges over all paths in $\graphone$ from $x_1$ to $x_2$, and $\height(\pi) = \max_{x\in \pi} f(x) - \min_{x\in \pi} f(x)$ is the maximum $f$-function value difference for points from the path $\pi$. 
Define the metric $d_g$ for $\graphtwo$ similarly. 

Now given two continuous maps $\leftmap: \graphone \to \graphtwo$ and $\rightmap: \graphtwo \to \graphone$, we consider the following \emph{continuous matching}, which is a correspondence induced by a pair of \emph{continuous maps} $\leftmap$ and $\rightmap$: 
\begin{align}
\cmatch(\leftmap, \rightmap) &= \{ (x, \leftmap(x)) \mid x \in \graphone \} \bigcup \{ (\rightmap(y), y) \mid y \in \graphtwo \}. 
\end{align}
The distortion induced by $\leftmap$ and $\rightmap$ is defined as: 
\begin{align}
\mathcal{D}(\leftmap, \rightmap) &= \sup_{(x,y), (x', y') \in \cmatch(\leftmap, \rightmap)} \frac{1}{2} |d_f(x,x') - d_g(y,y') |. 
\label{eqn:D}
\end{align}
The \emph{functional-distortion distance} between two metric graphs $(\graphone, d_f)$ and $(\graphtwo, d_g)$ is defined as: 
\begin{align}
\dfd(\graphone, \graphtwo) &= \inf_{\leftmap, \rightmap} \max \{ \mathcal{D}(\leftmap, \rightmap), \max_{x\in \graphone} |f(x) - g \circ \leftmap(x) |, \max_{y \in \graphtwo} |f\circ\rightmap(y) - g(y)| \},
\label{eqn:dfd}
\end{align}
where $\leftmap$ and $\rightmap$ range over all continuous maps between $\graphone$ and $\graphtwo$. 
It is shown in \cite{BGW14} that 
\begin{theorem}[\cite{BGW14}]\label{thm:fd}
$ d_B( \mathrm{Dg}_0 f, \mathrm{Dg}_0 g) \le \dfd((\graphone, d_f), (\graphtwo, d_g) )$, where $\mathrm{Dg}_0 f$ and $\mathrm{Dg}_0 g$ denote the 0th-dimensional persistence diagram induced by the function $f: \graphone \to \reals$ and $g: \graphtwo \to \reals$, respectively. 
\end{theorem}

In what follows, we show that $\dfd((\graphone, d_f), (\graphtwo, d_g) )\le 6\delta$ for $f = \done \bp$ and $g = \dtwo \nbp$. Note that in this case $\mathrm{Dg}_0 f = \perone \bp$ and $\mathrm{Dg}_0 g = \pertwo \nbp$. Combining with Theorem \ref{thm:fd}, this then implies $d_B(\perone \bp, \pertwo \nbp) \le 6\delta$. 

\paragraph*{Remark.}
We note that Theorem \ref{thm:fd} extends to the case where we consider the $1$st-extended persistence diagrams for $f$ and $g$, respectively, in which case the constant in front of $\dfd$ will change from $1$ to $3$. In other words, if we include the 1st-extended persistence diagrams in our definitions of the \spdist{} distance, then Theorem \ref{thm:stability} still holds with a slightly worst constant of $18$ (instead of $6$). 

\begin{lemma}\label{lem:functionalbound}
$\dfd((\graphone, d_{f}), (\graphtwo, d_{g}) )\le 6\delta$ for $f = \done \bp$ and $g = \dtwo \nbp$. 
\label{lem:dfbound}
\end{lemma}
\begin{proof}
First, we introduce the \emph{function-restricted Gromov-Hausdorff distance} $d_{rGH} ((\graphone, d_{G_1}), (\graphtwo, d_{G_2}))$, which is a more restricted version of the Gromov-Hausdorff distance, defined as follows: 
\begin{align}
d_{rGH} & ((\graphone, d_{G_1}), (\graphtwo, d_{G_2}))  \nonumber \\
&= \inf_{\mathcal{M}} \max \big\{ \frac{1}{2} \max_{(x_1,y_1), (x_2,y_2) \in \mathcal{M}} | d_{G_1}(x_1,x_2) - d_{G_2}(y_1,y_2)|, \max_{(x,y)\in \mathcal{M}}  |f(x) - g(y)| \} ~\big \}, \label{eqn:rGH}
\end{align}
where $\mathcal{M}$ ranges over all correspondences between graphs $G_1$ and $G_2$. 
Compared to the definition of the Gromov-Hausdorff distance in Eqn (\ref{eqn:dgh}), the functional Gromov-Hausdorff distance has an extra condition that the function value difference $|f(x) - g(y)|$ between a pair of corresponding points $x\in G_1$ and $y \in G_2$ should also be small. 

We claim that $\dfd((\graphone, d_f), (\graphtwo, d_g) )\le  3 d_{rGH} ((\graphone, d_{G_1}), (\graphtwo, d_{G_2}) )$. 
The proof follows almost exactly the same as the proof of Theorem A.1 of \cite{BGW14} (or Theorem 5.1 of the full arXiv version). Specifically, in Theorem A.1 of \cite{BGW14}, it states that $\dfd((\graphone, d_f), (\graphtwo, d_g) )\le  3 d_{fGH} ((\graphone, d_f), (\graphtwo, d_g) )$, where the so-called \emph{functional Gromov-Hausdorff distance} $d_{fGH} ((\graphone, d_f), (\graphtwo, d_g))$ defined as: 
\begin{align*}
d_{fGH} & ((\graphone, d_f), (\graphtwo, d_g))  \\
&= \inf_{\mathcal{M}} \max \big\{ \frac{1}{2} \max_{(x_1,y_1), (x_2,y_2) \in \mathcal{M}} | d_f(x_1,x_2) - d_g(y_1,y_2)|, \max_{(x,y)\in \mathcal{M}}  |f(x) - g(y)| \} ~\big \},
\end{align*}
where $\mathcal{M}$ ranges over all correspondences between graphs $G_1$ and $G_2$. 
In other words, the difference between $d_{fGH}$ and our $d_{rGH}$ is that the metric on $G_1$ (resp. on $G_2$) is $d_f$ versus the input graph metric $d_{G_1}$ (resp. $d_g$ versus $d_{G_2}$). 
Nevertheless, it turns out that the proof of $\dfd((\graphone, d_f), (\graphtwo, d_g) )\le  3 d_{fGH} ((\graphone, d_f), (\graphtwo, d_g) )$ can be easily modified to prove $\dfd((\graphone, d_f), (\graphtwo, d_g) )\le  3 d_{rGH} ((\graphone, d_{G_1}), (\graphtwo, d_{G_2}) )$. Specifically, one property that will be used many times in this modification is that $d_f(x_1, x_2) \le d_{G_1}(x_1,x_2)$ for any $x_1, x_2\in G_1$ (a symmetric statement holds for points in $G_2$). Since the proof is almost verbatim of the proof for Theorem A.1 in \cite{BGW14} (Theorem 5.1 of the full arXiv version), we omit it here. 


Given the optimal correspondence $\mathcal{M}^*$ for the Gromov-Hausdorff, it is easy to see that 
$|f(x) - g(y)| = |\done \bp(x) - \dtwo \nbp(y) | \le 2\delta$ for any pair $(x,y) \in \mathcal{M}^*$. 
For the correspondence $\mathcal{M}^*$, the other two terms in Equation \ref{eqn:rGH} are both bounded by $2\delta$. 
It then follows that 
$$ d_{rGH} ((\graphone, d_{G_1}), (\graphtwo, d_{G_2})) \le 2\delta. $$
Combining this with $\dfd((\graphone, d_f), (\graphtwo, d_g) )\le  3 d_{rGH} ((\graphone, d_{G_1}), (\graphtwo, d_{G_2}) )$, the lemma then follows. 
\end{proof}

We remark that one can in fact further modify the proof of Theorem A.1 of \cite{BGW14} to obtain a smaller constant in Lemma \ref{lem:functionalbound}. 

Putting everything together we obtain Theorem \ref{thm:stability}. 

\paragraph{A corollary of Lemma \ref{lem:functionalbound}.}
We note that in the case where the metric graphs $\graphone$ and $\graphtwo$ are trees, say we have two metric trees $(T_1, d_{T_1})$ and $(T_2, d_{T_2})$. If a tree, say $T_1$, is associated with a function $f: T_1 \to \reals$ such that the function value of $f$ is monotonically decreasing from the root to any leaf, then we also refer to $T_1$ equipped with the function a \emph{merge tree}, denoted by $T_{1f}$. (Similarly, denote by $T_{2g}$ a merge tree $T_2$ equipped with a function $g: T_2 \to \reals$.) 
Morozov et al. \cite{MBW13} introduces a so-called \emph{interleaving} distance for two merge trees, denoted by $d_I(T_{1f}, T_{2f})$. It 
is shown in \cite{BGW14} (Theorem 6.2 in the arXiv version) that $d_I(T_{1f}, T_{2g}) = \dfd((T_1, d_f), (T_2, d_g))$, where $d_f$ and $d_g$ are induced by $f$ and $g$ as defined in Eqn (\ref{eqn:df}), respectively. By Lemma \ref{lem:functionalbound} we then have the following result, which could be of independent interests.  
\begin{cor}\label{cor:treedistance}
Given two metric trees $(T_1, d_{T_1})$ and $(T_2, d_{T_2})$, let $\in T_1$ and $t\in T_2$ be such that $(s,t) \in \mathcal{M}^*$ is from an optimal correspondance $\mathcal{M}^*: T_1 \times T_2$ realizing the Gromov-Hausdorff distance between $(T_1, d_{T_1})$ and $(T_2, d_{T_2})$. Consider the functions $f = d_{T_1,s}: T_1 \to \reals$ and $g=d_{T_2,t}: T_2 \to \reals$ for base points $s\in T_1$ and $t\in T_2$. We then have that
$$d_I(T_{1f}, T_{2g}) = \dfd((\graphone, d_{f}), (\graphtwo, d_{g}) )\le 6 \dgh(T_1, T_2). $$
\end{cor}

\section{Discrete \spd-Distance} 
\label{sec:compdisc}

\newcommand{\midvertex}	{mid-vertex\xspace}
\newcommand{\midvertices}	{mid-vertices\xspace}
\newcommand{\mapone}		{\phi_1}
\newcommand{\maptwo}		{\phi_2}


Suppose we are given two connected metric graphs $(\graphone = (\Vone, \Eone), d_{G_1})$ and $(\graphtwo = (\Vtwo,\Etwo), d_{G_2})$, where the shortest distance metrics $d_{G_1}$ and $d_{G_2}$ are induced by lengths associated with the edges in $\Eone \cup \Etwo$. 
As a warm-up, we first consider the following \emph{discrete version} of \spdist{} distance where only graph nodes in $\Vone$ and $\Vtwo$ are used as base points: 
\begin{definition}\label{def:discretePD}
Let $\hat{\setone}:=\{ \perone{v} \mid v \in V(\graphone) \}$ and 
$\hat{\settwo}:=\{ \pertwo{u} \mid u \in V(\graphtwo) \}$ be 
two discrete sets of persistence diagrams.
The \emph{discrete \spdist{} distance} between $\graphone$ and $\graphtwo$, denoted by $\disdsp(\graphone,\graphtwo)$, is given by the Hausdorff distance
$d_H(\hat{\setone},\hat{\settwo})$. 
\end{definition}
We note that while we only consider graph nodes as base points, the local maxima of the resulting geodesic function may still occur in the middle of an edge. Nevertheless, for a fixed base point, each edge could have at most one local maximum, and its location can be decided in $O(1)$ time once the shortest-path distance from the base point to the endpoints of this edge are known. 
The observation below follows from the fact that geodesic distance is 1-Lipschitz (as the basedpoint moves) and from the stability of persistence diagrams. 
\begin{obs}\label{obs:discrete}
$\dsp(\graphone,\graphtwo) \le \disdsp(\graphone, \graphtwo) \le \dsp(\graphone, \graphtwo) + \frac{\ell}{2}$, where $\ell$ is the largest length of any edge in 
$\Eone\cup \Etwo$. 
\end{obs}

\begin{lemma}\label{lem:discretecomp}
Given connected metric graphs $\graphone = (\Vone, \Eone)$ and $\graphtwo = (\Vtwo,\Etwo)$, $\disdsp(\graphone, \graphtwo)$ can be computed in $O(n^2m^{1.5}\log m)$ time, where $n = \max \{|\Vone|, |\Vtwo|\}$ and $m = \max \{ |\Eone|, |\Etwo| \}$. 
\end{lemma}
\begin{proof}
For a given base point $\bp \in \Vone$ (or $\nbp \in \Vtwo$), computing the shortest path distance from $\bp$ to all other graph nodes, as well as the persistence diagram $\perone{\bp}$ (or $\pertwo{\nbp}$) takes $O(m\log  n)$ time. Hence it takes $O(mn \log n)$ total time to compute the two collections of persistence diagrams $\widehat{\setone}=\{ \perone{\bp} \mid \bp \in V(\graphone) \}$ and $\widehat{\settwo} = \{ \pertwo{\nbp} \mid \nbp \in V(\graphtwo) \}$. 

Each persistence diagram $\perone{\bp}$ has $O(m)$ number of points in the plane -- it is easy to show that there are $O(m)$ number of local maxima of the geodesic function $\done \bp$ (some of which may occur in the interior of graph edges). 
Since the birth time $\perb$ of every persistence point $(\perb, \perd)$ corresponds to a unique local maximum $\birthp$ with $f(\birthp) = \perb$, there can be only $O(m)$ points (some of which may overlap each other) in the persistence diagram $\perone{\bp}$. 

Next, given two persistence diagrams $\perone \bp$ and $\pertwo \nbp$, we need to compute the bottleneck distance between them. 
In \cite{EKI01}, Efrat \etal{} gives an $O(k^{1.5}\log k)$ time algorithm to compute the optimal bijection between two input sets of $k$ points $P$ and $Q$ in the plane such that the maximum distance between any mapped pair of points $(p, q) \in P \times Q$ is minimized. This distance is also called the bottleneck distance, and let us denote it by $\hat{d}_B$. 
The bottleneck distance between two persistence diagrams $\perone \bp$ and $\pertwo \nbp$ is similar to the bottleneck distance $\hat{d}_B$, with the extra addition of diagonals. However, let $P'$ and $Q'$ denote the vertical projection of points in $\perone \bp$ and $\pertwo \nbp$, respectively, onto the diagonal $L$. It is easy to show that $d_B(\perone,\pertwo) = \hat{d}_B(\perone\bp \cup Q', \pertwo \nbp \cup P')$. Hence $d_B(\perone \bp,\pertwo \nbp)$ can be computed by the algorithm of \cite{EKI01} in $O(m^{1.5}\log m)$ time. 
Finally, to compute the Hausdorff distance between the two sets of persistence diagrams $\widehat{\setone}$ and $\widehat{\settwo}$, one can check for all pairs of persistence diagrams from these two sets, which takes $O(n^2m^{1.5}\log m)$ time since the $|\widehat \setone| \le n$ and $|\widehat \settwo| \le n$. The lemma then follows. 
\end{proof} 

By Observation \ref{obs:discrete}, $\disdsp(\graphone, \graphtwo)$ only provides an approximation of $\dsp(\graphone,\graphtwo)$ with an \emph{additive error} as decided by the longest edge in the input graphs. For unweighted graphs (where all edges have length 1), this gives an additive error of $1$. As $\dsp(\graphone,\graphtwo)$ is necessarily an integer in this setting, this in turns provides a factor-2 approximation of the continuous \spdist{} distance in terms of \emph{multiplicative error}; see the following corollary. 
\begin{cor}
The discrete \spdist{} distance provides a factor-2 (multiplicative) approximation of the continuous \spdist{} distance for two graphs $\graphone$ and $\graphtwo$ with unit edge lengths; that is, 
$\dsp(\graphone, \graphtwo) \le \disdsp(\graphone, \graphtwo) \le 2 \dsp(\graphone, \graphtwo)$. 
\end{cor} 

One may add additional (steiner) nodes to edges of input graphs to reduce the longest edge length, so that the discrete \spdist{} distance approximates the continuous one within a smaller additive error. But it is not clear how to bound the number of steiner nodes necessary for approximating the continuous distance within a multiplicative error, even for the case when all edges weights are approximately 1. Indeed, even when all the edges of two input graphs have weights that are roughly 1, it is possible that the \spdist{} distance is much smaller than 1. 
Below we show how to directly compute the continuous \spdist{} distance \emph{exactly} in polynomial time.

\section{Computation of Continuous \Spdist{} Distance}
\label{sec:compcontinuous}


We now present a polynomial-time algorithm to compute the (continuous) \spdist{} distance between two metric graphs $(\graphone=(\Vone,\Eone), d_{G_1})$ and $(\graphtwo =(\Vtwo, \Etwo), d_{G_2})$. As before, set $n = \max\{ |\Vone|, |\Vtwo|\}$ and $m = \max\{ |\Eone|, |\Etwo|\}$. 
Below we first analyze how points in the persistence diagram change as we move the basepoint in $\graphone$ and $\graphtwo$ continuously. 

\subsection{Changes of persistence diagrams}
\label{subsec:perschanges}

We first consider the scenario where the basepoint $\bp$ moves within a fixed edge $\ebase \in \Eone$ of $\graphone$, and analyze how the corresponding persistence diagram $\perone{\bp}$ changes. 
Using notations from Section \ref{sec:notations}, 
let $(\birthp,\deathp)$ be the critical-pair in $\graphone$ that gives rise to the persistence point $(\perb,\perd)\in \perone{\bp}$. 
Then $\birthp$ is a maximum for the distance function $\done{\bp}$, while $\deathp$ is an up-fork saddle for $\done{\bp}$. We call $\birthp$ and $\deathp$ from $\graphone$ the \emph{birth point} and \emph{death point} w.r.t. the persistence-point $(\perb, \perd)$ in the persistence diagram $\perone{\bp}$. 

As the basepoint $\bp$ moves to $\bp' \in \ebase$ within $\eps$ distance along the edge $\ebase$ for any $\eps\ge 0$, the distance function is perturbed by at most $\eps$; that is, $\| \done{\bp} - \done{\bp'} \|_\infty \le \eps$. By the Stability Theorem of the persistence diagrams \cite{CEH07}, we have that $d_B(\perone{\bp}, \perone{\bp'}) \le\eps$. 
Hence as the basepoint $\bp$ moves continuously along $\ebase$, points in the persistence diagram $\perone{\bp}$ move continuously.
There could be new persistence points appearing or current points disappearing in the persistence diagram as $\bp$ moves. Both creation and deletion necessarily happen on the diagonal of the persistence diagram as $d_B(\perone{\bp},\perone{\bp'})$ necessarily tends to 0 as $\bp'$ approaches $\bp$. For simplicity of presentation, for the time being, we describe the movement of persistence points ignoring their creation and deletion. Such creation and deletion will be addressed later in Section \ref{subsec:tracking}. 

We now analyze how a specific point $(\perb,\perd)$ may change its trajectory as $\bp$ moves from one endpoint $v_1$ of $\ebase = (v_1, v_2) \in \Eone$ to the other endpoint $v_2$. 

Specifically, we use the arc-length parameterization of $\ebase$ for $\bp$, that is, $\bp: [0, \length(\ebase)] \to \ebase$.  
For any object $X \in  \{\perb, \perd, \birthp, \deathp \}$, we use $X(s)$ to denote the object $X$ w.r.t. basepoint $\bp(s)$. For example, $(\perb(s), \perd(s))$ is the persistence-point w.r.t. basepoint $\bp(s)$, while $\birthp(s)$ and $\deathp(s)$ are the corresponding pair of local maximum and up-fork saddle that give rise to $(\perb(s), \perd(s))$. 
We specifically refer to $\perb: [0, \length(\ebase)] \to \reals$ and $\perd: [0, \length(\ebase)] \to \reals$ as the \emph{birth-time function} and the \emph{death-time function}, respectively. By the discussion from the previous paragraphs on stability of persistence diagrams, these two functions are continuous. 

\begin{figure}[htp]
\begin{center}
\begin{tabular}{ccc}
\fbox{\includegraphics[height=1.65cm]{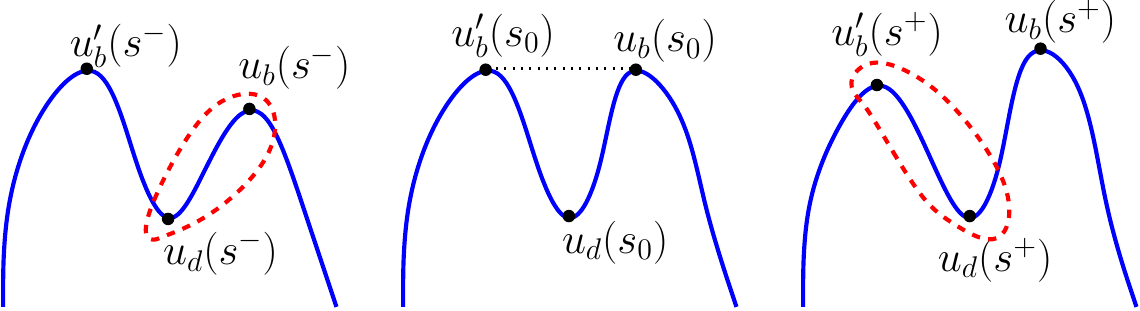}} &\hspace*{0in} & \fbox{\includegraphics[height=1.65cm]{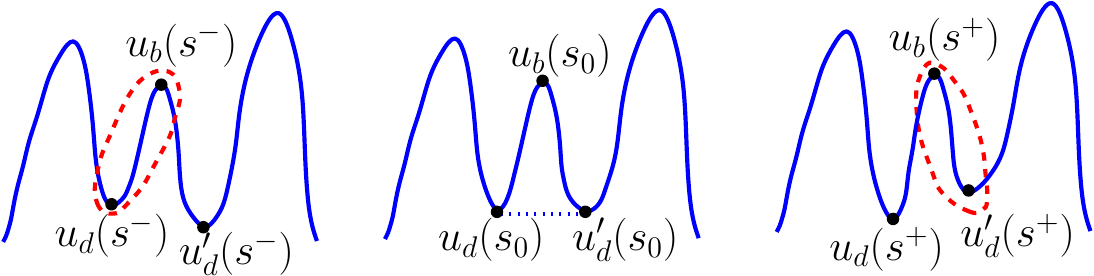}} \\
(a) & & (b) 
\end{tabular}
\end{center}
\vspace*{-0.2in}\caption{For better illustration of ideas, we use height function defined on a line to show: (a) a max-max critical event at $s_0$; and (b) a saddle-saddle critical event at $s_0$. 
\label{fig:criticalevents}}
\end{figure}
\paragraph{Critical events.} 
To describe the birth-time and death-time functions, we need to understand how the corresponding birth-point and death-point $\birthp(s)$ and $\deathp(s)$ in $\graphone$ change as the basepoint $\bp$ varies. 
Recall that as $\bp$ moves, the birth-time and death-time change continuously. However, the critical points $\birthp(s)$ and $\deathp(s)$ in $\graphone$ may 
(i) stay the same or move continuously, or (ii) have discontinuous jumps. 
Informally, if it is case (i), then we show below that we can describe $\perb(s)$ and $\perd(s)$ using a piecewise linear function with $O(1)$ complexity. 
Case (ii) happens when there is a \emph{critical event} where two critical-pairs ($u_b, u_d$) and $(u_b', u_d')$ swap their pairing partners to $(u_b, u_d')$ and $(u_b', u_d)$. 
Specifically, at a critical event, since the birth-time and death-time functions are still continuous, it is necessary that either $\done{\bp}(u_b)  = \done{\bp}(u_b')$ or $\done{\bp}(u_d) = \done{\bp}(u_d')$; we call the former a \emph{max-max critical event} and the latter a \emph{saddle-saddle critical event}. See Figure \ref{fig:criticalevents} for an illustration. 
It turns out that the birth-time function $\perb:  [0, \length(\ebase)] \to \mathbb{R}$ (resp. death-time function $\perd$) is a piecewise linear function whose complexity depends on the number of critical events, which we analyze below. 


\subsubsection{The death-time function $\perd: [0, \length(\ebase)] \rightarrow \mathbb{R}$}
\label{subsec:deathtime}


The analysis of death-time function is simpler than that of the birth-time function; so we describe it first. 
Observe that $\done \bp$ is the geodesic distance to the base point $\bp$. Consequently, merging of two components at an up-fork saddle cannot happen in the interior of an edge, unless at the basepoint $\bp$ itself.
\begin{obs}\label{obs:upfork}
An upper-fork saddle $u \in \graphone$ is necessarily a graph node from $\Vone$ with degree at least $3$ unless $u = \bp$. 
%
\end{obs}
\begin{wrapfigure}{r}{0.55\textwidth}
\begin{tabular}{ccc}
\fbox{\includegraphics[height=2cm]{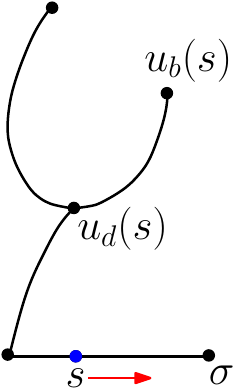}} & \fbox{\includegraphics[height=2cm]{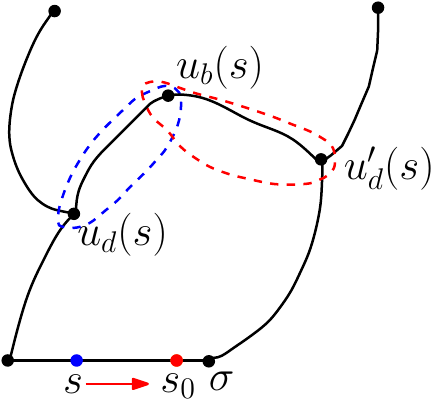}} &\fbox{\includegraphics[height=1.8cm]{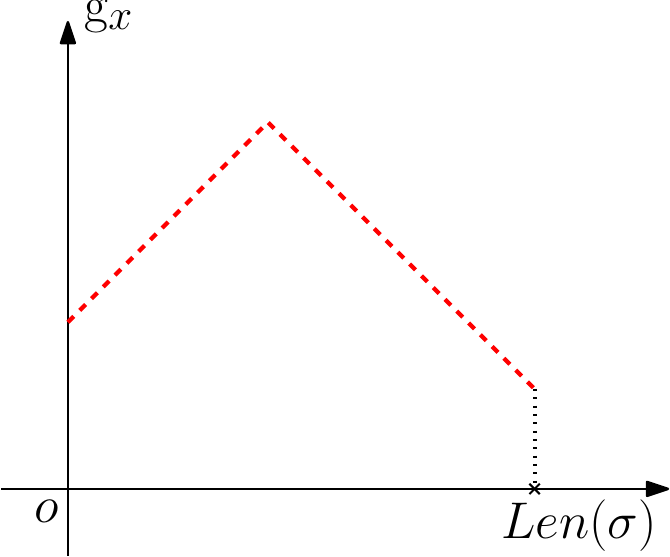}} \\
(case-1) & (case-2) & (c)
\end{tabular}
\vspace*{-0.1in}\caption{(c) Graph of function $\geod_x:[0,\length(\ebase)]\to \reals$. 
\label{fig:saddlecases}}
\end{wrapfigure}
To simplify the exposition, we omit the easier case of $u = \bp$ in our discussions below. 
Since the up-fork saddles now can only be graph nodes, as the basepoint $\bp(s)$ moves, the death-point $\deathp(s)$ either (case-1) stays at the same graph node, or (case-2) switches to a different up-fork saddle $u'_\perd$ (i.e, a saddle-saddle critical event); see 
Figure \ref{fig:saddlecases}. 

Now for any point $x\in \graphone$, we introduce the function $\geod_x: [0, \length(\ebase)]\to \mathbb{R}$ which is the distance function from $x$ to the moving basepoint $\bp(s)$ for $s\in  [0, L_\ebase]$; that is, $\geod_x(s) := \done{\bp(s)}(x)$. Intuitively, as the basepoint $\bp(s)$ moves along $\ebase$, the distance from $\bp(s)$ to a fixed point $x$ either increases or decreases at unit speed, until it reaches a point where the shortest path from $\bp(s)$ to $x$ changes discontinuously though the shortest path distance still changes continuously. We have the following observation. 

\begin{claim}\label{claim:distancetonode}
For any point $x\in \graphone$, as the basepoint $\bp$ moves in an edge $\ebase \in E$, the distance function $\geod_x: [0, \length(\ebase)] \rightarrow \mathbb{R}$ defined as $ \geod_x(s) := \done{\bp(s)}(x)$ is a piecewise linear function with at most 2 pieces, where each piece has slope either `1' or `-1'. See Figure \ref{fig:saddlecases} (c). 
\end{claim}
\begin{proof}
\label{appendix:claim:distancetonode}
\newcommand{\mym}		{{\mathbf{m}}}
Let $v_1$ and $v_2$ be the two endpoints of the edge $\ebase$ where the basepoint $\bp$ lies in. 
For a fixed point $x\in G_1$, first consider the shortest path tree $T_x$ with $x$ being the source point (root). 
If the edge $\ebase$ is a tree edge in the shortest path tree $T_x$, then as $\bp$ moves from $v_1$ to $v_2$, the shortest path from $x$ to $\bp$ changes continuously and the distance $d_{G_1}(x, \bp)$ increases or decreases at unit speed. In this case, the function $\geod_x$ contains only one linear piece with slope either `1' (if $\bp$ is moving towards $x$) or `-1' (if $\bp$ is moving away from $x$). 

Otherwise, the shortest distance to $\bp(s)$ from $x$ will be the shorter of the shortest distance to $v_i$ plus the distance from $v_i$ to $\bp(s)$, for $i = 1, $or $2$ and $s\in [0, \length(\ebase)]$.
That is, 
$$\geod_x(s) = \min \{ d_{G_1}(x, v_1) + s, d_{G_1}(x, v_2) + \length(\ebase) - s \}.$$
The two functions in the above equation are linear with slope `1' and `-1', respectively. The graph of $\geod_x$ is the lower envelop of the graphs of these two linear functions, and the claim thus follows. 

We note that the break point of the function $\geod_x$, where it changes to a different linear function, happens at the value $s_0$ such that $d_{G_1}(x, v_1) + s_0 = d_{G_1}(x, v_2) + \length(\ebase) - s_0$, and it is easy to check that $\bp(s_0)$ is a local  maximum of the distance function $\geod_x$. 
\end{proof}

As $\bp(s)$ moves, if the death-point $\deathp(s)$ stays at the same up-fork saddle $u$, then 
by the above claim, the death-time function $\perd$ (which locally equals $\geod_{u}$) is a piecewise linear function with at most 2 pieces. 

Now we consider (case-2) when a saddle-saddle critical event happens: Assume that as $s$ passes value $s_0$, $\deathp(s)$ switches from a graph node $u$ to another one $u'$.  
At the time $s_0$ when this swapping happens, we have that $\done{\bp(s_0)}(u)  = \done{\bp(s_0)}(u')$. In other words, the graph for function $\geod_u$ and the graph for function $\geod_{u'}$ intersect at $s_0$. 
Before $s_0$, the death function $\perd$ follows the graph for the distance function $\geod_u$, while after time $s_0$, $\deathp$ changes its identity to $u'$ and thus the movement of $\perd$ will then follow the distance function $\geod_{u'}$ for $s > s_0$. Since the function $\geod_x$ is piecewise-linear (PL) with at most $2$ pieces as shown in Figure \ref{fig:saddlecases} (c) for any point $x \in \graphone$, the switching for a fixed pair of nodes $u$ and $u'$ can happen at most once (as the graph of $\geod_u$ and that of $\geod_{u'}$ intersect at most once). 
Overall, since there are $|\Vone| \le n$ graph nodes, we conclude that: 

\begin{lemma}
As $\bp$ moves along $\ebase$, there are $O(n^2)$ number of saddle-saddle critical events in the persistence diagram $\perone{\bp}$. 
\end{lemma}

For our later arguments, we need a stronger version of the above result. Specifically, imagine that we track the trajectory of the death-time $\perd$ for a persistence pair $(\perb, \perd)$. 
\begin{proposition}\label{prop:tracksaddle}
For a fixed persistent point $(\perb(0), \perd(0)) \in \perone{\bp(0)}$, the corresponding death-time function $\perd: [0, \length(\ebase)] \to \mathbb{R}$ is piecewise linear with at most $O(n)$ pieces, and each linear piece has slope either `1' or `-1'. 
This also implies that the function $\perd$ is 1-Lipschitz. 
\end{proposition}
\begin{proof} %
%
By Observation \ref{obs:upfork}, $\deathp(s)$ is always a graph node from $\Vone$. 
For any node $u$, recall $\geod_u(s) = \done{\bp(s)}(u)$. 
As described above, $\perd(s)$ will follow certain $\geod_u$ with $u = \deathp(s)$ till the identify of $\deathp(s)$ changes at a saddle-saddle critical event between $u$ with another up-fork saddle $u'$. Afterwards, $\perd(s)$ will follow $\geod_{u'}$ till the next critical event. 
Since each piece of $\geod_v$ has slope either `1' or `-1', the graph of $\perd$ consists of linear pieces of slope `1' or `-1'. 
Note that this implies that the function $\perd$ is a 1-Lipschitz function. 

On the other hand, for a specific graph node $u \in V$, each linear piece in $\geod_u$ has slope `1' or `-1'. This means that one linear piece in $\geod_u$ can intersect the graph of $\perd$ at most once for $s\in [0, \length(\ebase)]$ as $\perd$ is 1-Lipschitz. Hence the graph of $\geod_u$ can intersect the graph of $\perd$ at most twice; implying that the node $u$ can appear as $\deathp(s)$ for at most two intervals of $s$ values. Thus the total descriptive complexity of $\perd$ is $O(|\Vone|) = O(n)$, which completes the proof. 
\end{proof}

\newcommand{\mym}		{{\mathbf{m}}}

\subsubsection{The Birth-time Function $\perb: [0, \length(\ebase)] \to \mathbb{R}$} 
\label{appendix:subsec:birthtime}

To track the trajectory of the birth-time $\perb$ of a persistence pair $(\perb(0), \perd(0)) \in \perone{0}$, we study the movements of its corresponding birth-point (which is a maximum) $\birthp: [0, \length(\ebase)] \to G$ in the graph. 
However, unlike up-fork saddles which must be graph nodes, maxima of the distance function $\done{\bp}$ can also appear in the interior of a graph edge. 
Roughly speaking, in addition to degree-1 graph nodes, which must be local maxima, 
imagine the shortest path tree with $\bp$ as the root (source), then for any non-tree edge, it will generate a local maximum of the distance function $\done{\bp}$. (Recall the maximum $u$ in Figure \ref{fig:graphexample} (b), which lies in the interior of edge $(v_3, v_4)$. ) Nevertheless, the following result states that there can be at most one local maximum associated with each edge. 

\begin{lemma}\label{lem:max}
Given an arbitrary basepoint $\bp$, a maximum for the distance function $\done{\bp}: \graphone \to \mathbb{R}$ is either a degree-1 graph node, or a point $v$ with at least two shortest paths to the basepoint $\bp$ which are disjoint in a small neighborhood around $v$. 

Furthermore, there can be at most one maximum of $\done{\bp}$ in each edge in $\Eone$. 
\end{lemma}
\begin{proof}
Consider the shortest path tree $T$ of $\graphone$ rooted at $\bp$. 
All degree-1 graph nodes in $V$ will be tree leaves, and each of them is thus a local maximum for the distance function $\done{\bp}$. 
For each such maximum, we associate it with the unique tree edge incident on it. 

Now take a maximum $v$ which is not a degree-1 graph node. Set $k$ to be the number of branches incident on $v$ in a sufficiently small neighborhood of $v$: $k = 2$ if $v$ is in the interior of an edge of $\Eone$, and $k= degree(v) \ge 2$ if $v$ is a graph node. 
Since $\done{\bp}(v)$ is larger than the distance from basepoint $\bp$ to any other point in the neighborhood of $v$, and since the distance function is continuous, there must exist at least $k$ different shortest paths from $\bp$ to $v$, each one coming from a different branch around $v$ (and thus disjoint in a small neighborhood around $v$). 

Furthermore, for each of the $\Eone-\Vone+1$ number of edges not in the shortest path tree $T$ rooted at $\bp$, say $e = (w_1,w_2)$, it must contain one local maximum for the distance function $\done{bp}$. 
Indeed, by property of shortest path distance, we know that $\done{\bp}(w_1) \le \done{\bp}(w_2) + \length(e)$ and $\done{\bp}(w_2) \le \done{\bp}(w_1) + \length(e)$. 
If the equality does not hold in either of these two relations, then as we move $x$ from the endpoint with lower distance value, say $w_1$, to $w_2$ along $e$, the shortest distance must first increase and then decrease, meaning that there is a local maximum in the interior of $e$. 
Specifically, the local maximum happens at the point $v\in e$ such that $\done{\bp}(w_1) + \|w_1 - v \| = \done \bp(w_2) + \| v - w_2 \|$, and there are two shortest paths from $\bp$ to $v$, one passing through $w_1$ and the other passing through $w_2$. 
If the equality holds for one of them, say $\done{\bp}(w_2) = \done{\bp}(w_1) + \length(e)$, then $w_2$ may or may not be a local maximum. 

Overall, each edge in $G$, whether it is a tree edge or non-tree edge in $T$, will produce at most one local maximum for the distance function $\done{\bp}$.  The claim the follows. 
\end{proof}

 As the basepoint $\bp$ moves, the position of the local maximum within an edge may stay or may move \emph{continuously} along the edge $e$. The above claim states that for a fixed basepoint, there can be at most one maximum in an edge $e\in \Eone$. 
Hence instead of tracking $\birthp$ (which could move continuously), we now associate the identity of $\birthp$ with the \emph{birth-edge} $\deathe \in \Eone$ that contains $\birthp$, and track the changes of the birth-edge $\deathe:[0, \length(\ebase)] \to \Eone$ as the basepoint moves: In particular, as $s\in [0, \length(\ebase)]$ changes, 
$\deathe(s)$ can remain as the same edge, or it can change to a different one. We now investigate each of these two cases.

\paragraph{Case 1: $\deathe$ does not change.}
\begin{wrapfigure}{r}{0.2\textwidth}
\includegraphics[height=2cm]{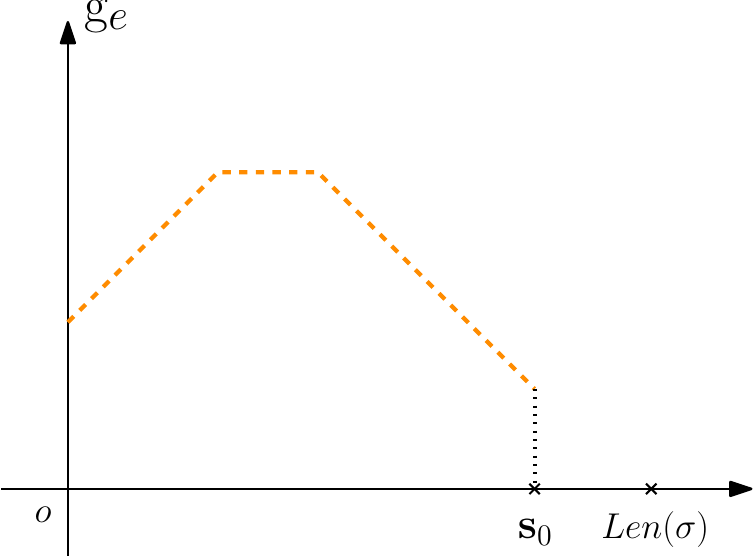}
\caption{Graph of function $\geode_e: [0, \length(\ebase)] \to \reals$.
\label{fig:graphsofdistance-max}}
\end{wrapfigure}
For a fixed edge $e\in \Eone$ we introduce the function $\geode_e: [0, \length(\ebase)] \to \mathbb{R}$ where, for any $s\in [0, \length(\ebase)]$, $\geode_e(s)$ is the distance from the basepoint $\bp(s)$ to the unique maximum (if it exists) in $e$; $\geode_e(s) = +\infty$ if the distance function $\done{\bp(s)}$ does not have a local maximum in $e$. We refer to the portion of $\geode_e$ with finite value as \emph{well-defined}.  
Intuitively, the function $\geode_e$ serves as the same role as the distance function $\geod_x$ in Section \ref{subsec:deathtime}, and similar to Claim \ref{claim:distancetonode}, we have the following characterization for this distance function. 

\begin{proposition}\label{prop:maxdistance}
For any edge $e\in \Eone$, the well-defined portion of the function $\geode_e$ is a piecewise-linear function with $O(1)$ pieces, where each piece is of slope `1', `-1' or `0'. See Figure \ref{fig:graphsofdistance-max} for an illustration. 
\end{proposition}
\begin{proof}
We assume that $e \neq \ebase$; the case $e = \ebase$ is simpler to handle. 
Let $\mym(s) \in e = (w_1, w_2)$ be the maximum for distance function $\done{\bp(s)}$ w.r.t. basepoint $\bp(s) \in \ebase = (v_1, v_2)$. Note that $\geode_e(s) = \done{\bp(s)}(\mym(s))$. 
From the proof of Lemma \ref{lem:max}, we know that 
\begin{align} 
\done{\bp(s)}(\mym(s)) &= \done{\bp(s)}(w_1) + \| w_1 - \mym(s)\| = \done{\bp(s)}(w_2) + \| w_2 - \mym(s) \| \\
&= \geod_{w_1}(s) + \| w_1 - \mym(s)\| = \geod_{w_2}(s) + \| w_2 - \mym(s) \|. 
\label{eqn:interiormax}
\end{align}
Recall from Section \ref{subsec:deathtime} that $\geod_x: [0,\length(\ebase)]\to \reals$ is defined as $\geod_x (s) = \done{\bp(s)}(x)$. 
Conversely, a point $\mym(s)$ in the interior of $e$ satisfying the equation above must be a local maximum of the distance function $\done{\bp(s)}$. 
By Claim \ref{claim:distancetonode}, as $s$ varies, $\geod_{w_1}$ (resp. $\geod_{w_2}$) is a piecewise linear function with at most two pieces of slope `1' or `-1'. 

(1) If at $s \in [0, \length(\ebase)]$, the slopes of functions $\geod_{w_1}$ and $\geod_{w_2}$ are the same (i.e, as $s$ increases, $\done{\bp(s)}(w_1)$ and $\done{\bp(s)}(w_2)$  both increase or both decrease at the same speed), 
then by Eqn (\ref{eqn:interiormax}), the local maximum $\mym(s)$ remains the same as $s$ moves. Hence  $\geode_e(s) = \done{\bp(s)}(\mym(s))$ follows a linear function with the same slope as $\geod_{w_1}$ (and $\geod_{w_2}$) which is either `1' or `-1'. 

(2) If at $s$, the slopes of $\geod_{w_1}$ and $\geod_{w_2}$ are not the same, i.e, as $s$ increases, $\done{\bp(s)}(w_1)$ and $\done{\bp(s)}(w_2)$ change in the opposite directions, then in order for Eqn (\ref{eqn:interiormax}) to hold, $\mym(s)$ moves at the same speed as $\bp(s)$. In this case, $\geode_e(s)$ remains the same value, that is, $\geode_e$ is a linear (in fact, constant) function with slope `0'. 

Now decompose $[0, \length(\ebase)]$ into maximal intervals such that within each interval, $\geod_{w_1}$ and $\geod_{w_2}$ each can be described by a single linear function. Due to the shape of the graph of $\geod_{w_1}$ and $\geod_{w_2}$ as specified by Claim \ref{claim:distancetonode}, there can be at most three such intervals. Within each interval, if a maximum exists in edge $e$, then the function $\geode_e$ (which is the distance to this maximum) can be described by a linear function of slope `1', `-1' or `0', as described by the two cases above. 

Finally, note that in (2) above, as $\mym(s)$ moves along $e$, it is possible that $\mym(s)$ coincides with one of its endpoint say $w_1$. After that, Eqn (\ref{eqn:interiormax}) cannot hold 
and the local maximum moves out of edge $e$ -- Indeed, one can verify that after that, the edge $e$ becomes a tree edge in the shortest path tree rooted at $\bp(s)$. In other words, afterwards, $\geode_e$ is no longer well-defined. 
Within a single maximal interval of $[0, \length(\ebase)]$ as described above, such event can happen at most once for each of $w_1$ and $w_2$. 
Overall, the well-defined portion of $\geode_e$ consists of $O(1)$ linear functions of slope `1', `-1' or `0'. 
\end{proof}

We remark that we can actually obtain a stronger characterization for the function $\geode_e$, which states that the well-defined portion has to be connected, and consists of at most three pieces with a graph as shown in Figure \ref{fig:graphsofdistance-max} (any piece can be degenerate). However, the above proposition suffices for our later arguments. 

\paragraph{Case 2: $\deathe$ changes from edge $e$ to $e'$. }
The change of the identity of $\deathe$ could be due to that the local  maximum $\birthp$ 
moves continuously from $e$ to a neighboring edge $e'$ that shares an endpoint with $e$. Alternatively, it could be caused by a max-max type critical event: Specifically, let $\deathp$ be the up-fork saddle currently paired with the current birth point $\birthp = u \in \deathe$ generating the birth time $\perb$ of $(\perb, \perd)$ in the persistence diagram. At a max-max critical event, the up-fork saddle changes its pairing partner from $\birthp = u \in \deathe$ to another maximum $u'$ in edge $e'$. Afterwards, the identify of $\deathe$ corresponding to the birth-time $\perb$ will change to $e'$. 
%
At the time $s_0$ when the swapping happens, $\done{\bp(s_0)}(u) = \done{\bp(s_0)}(u')$. It then follows that $\geode_e(s_0) = \geode_{e'}(s_0)$; that is, $s_0$ corresponds to an intersection point between the graph of the function $\geode_e$ and that of the function $\geode_{e'}$. 
Since the function $\geode_e$ consists of $O(1)$ linear pieces for any $e$, there are $O(1)$ intersection points between a pair of $e$ and $e'$ from $\Eone$. We thus have: 
\begin{lemma}
There are $O(m^2)$ max-max critical events as the basepoint $\bp$ moves along a fixed edge $\ebase\in E$. 
\end{lemma}

As in the case of tracking the death-time function $\perd$, our later analysis requires a stronger result bounding the descriptive complexity of the birth-time function $\perb: [0, \length(\ebase)]\to\mathbb{R}$, starting from a birth-time $\perb(0)$ from a fixed persistence pair $(\perb(0), \perd(0))\in \perone{\bp(0)}$. In particular, we have the following proposition: 
\begin{proposition}\label{prop:trackmax}
For a fixed $(\perb(0), \perd(0)) \in \perone{\bp(0)}$, the birth-time function $\perb: [0, \length(\ebase)] \to \mathbb{R}$, tracking birth-time $\perb(0)$, is piecewise linear with at most $O(m)$ pieces, and each linear piece has slope either `1', `-1', or `0'. 
Note that this also implies that the function $\perb$ is 1-Lipschitz. 
\end{proposition}
\begin{proof}
We track the edge $\deathe(s)$ containing the maximum $\birthp(s)$ that gives rise to the birth-time $\perb(s)$ for $s \in [0, \length(\ebase)]$. 
As described above, $\perb(s)$ will follow $\geode_e$ for $e = \deathe(s)$ till $\deathe$ changes its identity to a new edge $e'$. Afterwards, $\perb(s)$ will follow $\geode_{e'}$ till next time $\deathe$ changes identity. By Proposition \ref{prop:maxdistance}, $\perb$ thus consists of a set linear linear functions, each of slope `1', `-1', or `0'. Note that this also implies that $\perb$ is a 1-Lipschitz function. 

We now bound the descriptive complexity of $\perb$. 
Note that any break-point between two consecutive linear pieces in $\perb$ of different slopes necessarily involve at least one linear piece of slope either `1' or `-1'. So we can charge the number of break-points to the number of non-constant linear pieces in $\perb$. 
On the other hand, consider any non-constant piece from the function $\geode_e$: This piece can appear in the graph of the function $\perb$ at most once, because $\perb$ is 1-Lipschitz, and $\geode_e$ has slope either `1' or `-1'. 
Since there are $O(m)$ edges in $\Eone$, there are $O(m)$ non-constant linear pieces from all functions $\geode_e$, with $e\in \Eone$, which implies that there are only $O(m)$ number of breakpoints in $\perb$. This proves the lemma. 
\end{proof}

\paragraph{Remark.}
The readers may have the following question: Recall that the function $\geode_e$ could contain portions which are not well-defined. 
Suppose at some point, $\deathe = e$ and $\perb$ is following the graph of $\geode_e$. 
What if we reach the endpoint $s_0$ of the well-defined portion of $\geode_e$? We note that when this happens, as detailed in the proof of Proposition \ref{prop:maxdistance}, the corresponding maximum $\birthp$ currently is an endpoint say $w_1$ of $e$, and as the basepoint continues to change, either, $\birthp(s)$ moves to a neighboring edge $e'$ of $e$ incident on $w_1$; or, $w_1$ was a up-saddle prior to $s_0$ and at time $s_0$, the max $\birthp=w_1$ cancel with $\deathp = w_1$ (which we describe in Section \ref{subsec:tracking} below). Overall, as the Stability Theorem guarantees, $\perb$ is necessarily a continuous function. 

\subsubsection{Tracking the persistence pair $(\perb, \perd): [0, \length(\ebase)] \to \mathbb{R}^2$. }
\label{subsec:tracking}

\begin{wrapfigure}{r}{0.2\textwidth}
\includegraphics[height=3cm]{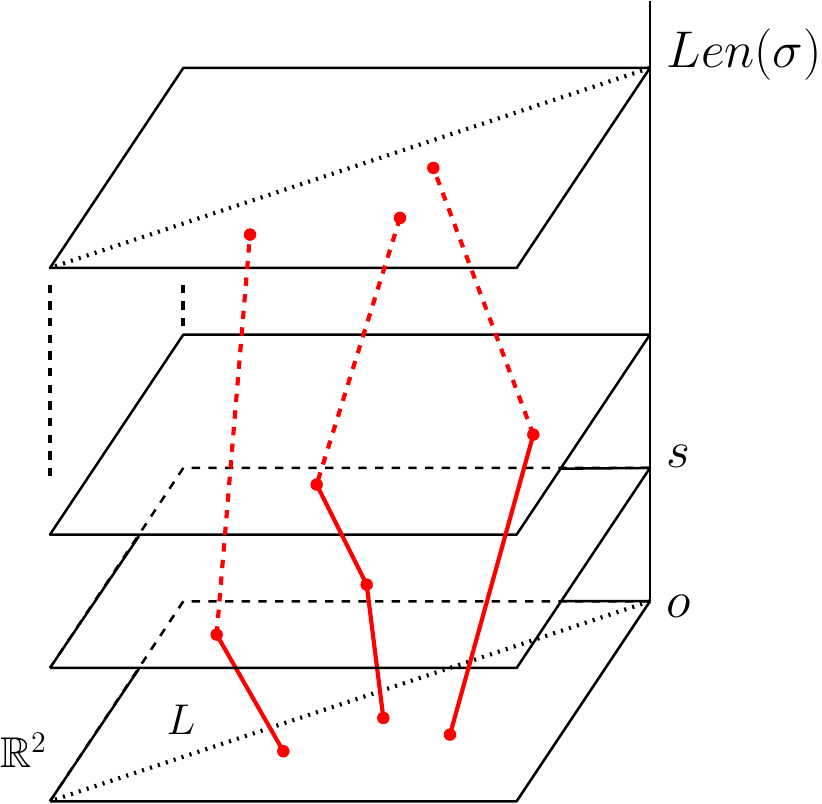}
\end{wrapfigure}
Now consider the space $\pertube_\ebase:= [0, \length(\ebase)] \times \mathbb{R}^2$, where $\mathbb{R}^2$ denotes the birth-death plane: We can think of $\pertube_\ebase$ as the stacking of all the planes containing persistence diagrams $\perone{\bp(s)}$ for all $s\in [ 0, \length(\ebase)]$. Hence we refer to $\pertube_\ebase$ as the \emph{stacked persistence-space}. 
For a fixed persistence pair $(\perb , \perd) \in \perone{\bp(s)}$, as we vary $s\in [0, \length(\ebase)]$, it traces out a \emph{trajectory} $\traj = \{ (s, \perb(s), \perd(s)) \mid s\in [0, [\length(\ebase)] \} \in \pertube_\ebase$, which is the same as the ``vines'' introduced by Cohen-Steiner \etal{} \cite{CEM06}. 
By Propositions \ref{prop:tracksaddle} and \ref{prop:trackmax}, the trajectory $\pi$ is a polygonal curve with $O(n+m) = O(m)$ linear pieces. See the right figure for an illustration, where there are three trajectories in the stacked persistence diagrams.


\begin{theorem}\label{thm:trajectories}
Let $\ebase \in \Eone$ be an arbitrary edge from the metric graph $(\graphone, d_{G_1})$. 
As the basepoint $\bp$ moves from one endpoint to another endpoint of $\ebase$ by $\bp: [0, \length(\ebase)] \to \ebase$, the persistence-points in the persistence diagram $\perone{\bp(s)}$ of the distance function $\done{\bp(s)}$ form $O(m)$ number of trajectories in the stacked persistence-space $\pertube_\ebase$. Each trajectory is a polygonal curve of $O(m)$ number of linear segments. 

A symmetric statement holds for the metric graph $(\graphtwo, d_{G_2})$. 
\end{theorem}
\paragraph{Proof of Theorem \ref{thm:trajectories}.}
%
As the basepoint $\bp(s)$ moves along an edge $\ebase$ with $s\in [0, \length(\ebase)]$, we can think of the distance function $\done{\bp(s)}$ as a time-varying function with time range $[0, \length(\ebase)]$. 
For a general time-varying function, as we track a specific point in the persistence diagram \cite{CEM06}, it is possible that the pair of critical points giving rise to this persistent-point may coincide and cease to exist afterwards. 
In this case, the corresponding trajectory (vine)  hits the diagonal of the persistence diagram (since as the two critical points coincide with $\birthp = \deathp$, we have that $\perb = \perd$) and terminates. 
The inverse of this procedure indicates the creation of a new trajectory.  
Hence a trajectory in the stacked persistence-diagrams may not span the entire range $[0, \length(\ebase)]$. 

We claim that there can be only $O(m)$ number of trajectories in the stacked persistence diagram. 
In particular, first, note that at time $s=0$, there can be $O(n+m) = O(m)$ number of persistence-points in the persistence diagram $\perone{\bp(0)}$ for basepoint $s(0)$. 
This is because that for a fixed basepoint, by Lemma \ref{lem:max}, there can be only $O(n+m)$ number of local maxima for the distance function $\done{\bp(0)}: \graphone \to \reals$, thus generating $O(m)$ number of persistence-points in the persistence diagram. 
As a result, there can be at most $O(m)$ trajectories spanning through the entire range $[0, \length(\ebase)]$. 

We next bound the number of trajectories not spanning the entire range. That is, these are the trajectories created or terminated at some time in $(0, \length(\ebase))$. 
For any such trajectory, assume without loss of generality that it tracks a persistence-point $(\perb, \perd)$, and terminates at time $s_0$. (The case of creation of a new trajectory is symmetric.) 
At this point, it is necessary that the two critical points $\birthp$ (a local maximum) and $\deathp$ (a up-fork saddle) coincide. 
By Observation \ref{obs:upfork}, the death-point $\deathp$ must be a graph node, say $w_0 \in \Vone$.  Hence $\birthp = w_0$ as well; that is, $w_0$ is also a maximum of the distance function $\done{\bp(s_0)}$. 
We show that for a fixed graph node $w_0$, such a scenario can happen at most once. 
\begin{lemma}
For a fixed graph node $w_0\in \Vone$, the birth-point $\birthp$ and death-point $\deathp$ can coincide at $w_0$ at most once as $s$ varies in the range $[0, \length(\ebase)]$. 
\end{lemma}
\begin{proof}
As above, assume the trajectory hits the diagonal of the persistence diagram at time $s_0$, and $w_0$ is the corresponding coincided birth- and death-points. 
Suppose at $s^- < s_0$ infinitesimally close to $s_0$, the corresponding local maximum $x^- = \birthp(s^-)$ comes from edge $e$ incident on $w_0$. 
Assume without loss of generality that $s^-$ is sufficiently close to $s_0$ such that there is no critical event of any kind and the local maximum $\birthp(s)$ approach continuously to $w_0$ as $s$ tends to $s_0$ (i.e, for $s\in (s^-, s_0)$). 

\parpic[r]{\includegraphics[height=2.7cm]{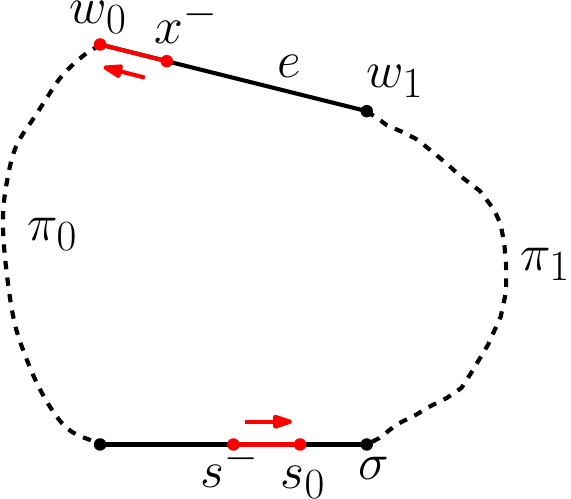}}
Let $w_1$ be the other endpoint of $e$. 
Since $x^-$ is a maximum of $\done{\bp(s^-)}$, by Lemma \ref{lem:max}, there are two shortest paths from $\bp(s^-)$ to $x^-$ passing through $w_0$ and $w_1$, which we denote by $\pi_0$ and $\pi_1$, respectively. 
We show that $\pi_0$ and $\pi_1$ in fact are disjoint other than at their endpoints $x^-$ and $\bp(s^-)$. 
See the right figure for an illustration.

Indeed, consider the shortest path tree $T^-$ rooted at $\bp(s^-)$, and let $z$ be the common ancestor of $w_0$ and $w_1$; $z$ is necessarily a graph node of $\Vone$ unless $z=\bp(s^-)$. 
If $z$ is a graph node, then as $s$ varies, the distance to $z$ either increases or decreases. However, the shortest path distance from $z$ to $w_0$ and to $w_1$ remain the same. 
Hence either both distance functions $\geode_{w_0} = \done{\bp(s)}(w_0)$ and $\geode_{w_1} = \done{\bp(s)}(w_1)$ increase or both decrease because the shortest distance to them is the shortest distance to $z$ plus the shortest distance from $z$ to each of them. 
However, this falls into case (1) in the proof of Proposition \ref{prop:maxdistance}, which means that the local maximum necessarily remains the same at $x^-$ as $s$ moves from $s^-$ to $s_0$, and will not move to $w_0$. 
Contradiction. As such, $z$ must be $\bp(s^-)$. 
In other words, the two shortest paths $\pi_1$ and $\pi_2$ meet only at $\bp(s^-)$ and $x^-$: Their concatenation form a simple loop $C$ where $x^-$ and $\bp(s^-)$ are a pair of \emph{antipodal points} along this loop (i.e, they bisect $C$). 
As $s$ moves to $s_0$, its corresponding local maximum $\birthp(s)$ remains the antipodal point of $\bp(s)$ and moves towards $w_0$, and $w_0$ is the antipodal point of $\bp(s_0)$. 

In other words, let $v_1, v_2$ denote the two endpoints of the edge $\ebase$ where the basepoint $\bp$ lies in. Since $w_0$ is the antipodal point of $\bp(s_0)$, we have that $d_{G_1}(v_1, w_0) + s_0 = d_{G_1}(v_2, w_0) + \length(\ebase) - s_0$. Hence there is only one possible value of $s_0$ for a fixed graph node $w_0$. 
This proves the lemma.
\end{proof}

It then follows that there can be at most $O(n)$ number of trajectories not spanning the entire time range $[0, \length(\ebase)]$ (created or terminated in the stacked persistence diagrams). 
Putting everything together, we have that there are at most $O(n+m)=O(m)$ trajectories in the stacked persistence diagrams as the basepoint $\bp$ moves in an edge $\ebase \in \Eone$. Combining this with Propositions \ref{prop:tracksaddle} and \ref{prop:trackmax}, Theorem \ref{thm:trajectories} then follows. 

\subsection{Computing $\dsp(\graphone,\graphtwo)$ }
\label{subsec:computation}

Given a pair of edges $\ebase_s \in \graphone$ and $\ebase_t\in \graphtwo$, as before, we parameterize the basepoints $\bp$ and $\nbp$ by the arc-length parameterization of $\ebase_s$ and $\ebase_t$; that is: $\bp: [0, L_s] \to \ebase_s$ and $\nbp: [0, L_t] \to \ebase_t$ where $L_s = \length(\ebase_s)$ and $L_t = \length(\ebase_t)$. We now introduce the following function to help compute $\dsp(\graphone,\graphtwo)$: 
\begin{definition}\label{def:bottleneckdist}
The \emph{bottleneck distance function} $\onePD_{\ebase_s,\ebase_t}:\recR \to \reals$ is defined as $\onePD_{\ebase_s,\ebase_t} (s,t) \mapsto d_B(\perone{\bp(s)}, \pertwo{\bp(t)})$. For simplicity, we sometimes omit $\ebase_s,\ebase_t$ from the subscript when their choices are clear from the context. 
\end{definition}
Recall that $\setone = \{ \perone{\bp} \mid \bp \in \graphone \}$, $\settwo = \{ \pertwo{\nbp} \mid \nbp \in \graphtwo\}$, and by Definition \ref{def:spdist}: 
$$\dsp(\graphone, \graphtwo) = \max \{ \max_{\mathrm{P} \in \setone} \min_{\mathrm{Q} \in \settwo} d_B(\mathrm{P}, \mathrm{Q}), ~ \max_{\mathrm{P} \in \settwo} \min_{\mathrm{P} \in \setone} d_B(\mathrm{P}, \mathrm{Q}) ~\}. $$
Below we focus on computing $\vec{d}_H(\setone,\settwo):=\max_{\mathrm{P} \in \setone} \min_{\mathrm{Q} \in \settwo} d_B(\mathrm{P}, \mathrm{Q})$, and the treatment of $\vec{d}_H(\settwo,\setone):=\max_{\mathrm{P} \in \settwo} \min_{\mathrm{P} \in \setone} d_B(\mathrm{P}, \mathrm{Q})$ is symmetric. 
It is easy to see: 
\begin{equation} \label{eqn:PDthree}
\vec{d}_H(\setone,\settwo)= \max_{\mathrm{P} \in \setone} \min_{\mathrm{Q} \in \settwo} d_B(\mathrm{P}, \mathrm{Q}) = \max_{\ebase_s \in \graphone} ~\max_{s \in [1, L_s]}~ \min_{\ebase_t \in \graphtwo}~\min_{t\in [1, L_t]} \onePD_{\ebase_s, \ebase_t} (s, t). 
\end{equation}

\noindent In what follows, we present the descriptive complexity of $\onePD_{\ebase_s,\ebase_t}$ for a fixed pair of edges $\ebase_s \in \graphone$ and $\ebase_t\in \graphtwo$ in Section \ref{subsec:onepair}, and show how to use it to compute the \spdist{} in Section \ref{subsec:altogether}. 

\subsubsection{One pair of edges $\ebase_s\in \graphone$ and $\ebase_t\in \graphtwo$.}
\label{subsec:onepair}

Recall that we call the plane containing the persistence diagrams as the \birthdeath{} plane, and for persistence-points in this plane, we follow the literature and measure their distance under the $L_\infty$-norm (recall Definition \ref{def:bottleneck}). 
From now on, we refer to persistence-points in $\perone{\bp(s)}$ as \emph{red points}, while persistence-points in $\pertwo{\nbp(t)}$ as \emph{blue points}. 
As $s$ and $t$ vary, the red and blue points move in the \birthdeath{} plane. By Theorem \ref{thm:trajectories}, the movement of each red (or blue) point traces out a polygonal curve with $O(m)$ segments (which are the projections of the trajectories from the stacked persistence diagrams onto the \birthdeath{} plane).  

Set $\recR:=[0, L_s] \times [0, L_t]$ and we refer to it as the \emph{s-t domain}. 
For a point $(s,t) \in \recR$, the function value $\onePD(s,t) (= \onePD_{\ebase_s,\ebase_t}(s,t)) = d_B( \perone{\bp(s)}, \pertwo{\nbp(t)})$ 
is the bottleneck distance between the set of red and the set of blue points (with the addition of diagonals) in the \birthdeath{} plane. 
To simplify the exposition, in what follows we ignore the diagonals from the two persistence diagrams and only consider the bottleneck matching between red and blue points. 

Let $r^*(s) \in \perone{\bp(s)}$ and $b^*(t) \in \pertwo{\nbp(t)}$ be the pair of red-blue points from the bottleneck matching between $\perone{\bp(s)}$ and $\pertwo{\nbp(t)}$ such that $d_\infty(r^*(s), b^*(t)) = d_B( \perone{\bp(s)}, \pertwo{\nbp(t)})$. We call $(r^*(s), b^*(t))$ \emph{the bottleneck pair (of red-blue points) w.r.t. $(s,t)$}. 
As $s$ and $t$ vary continuously, red and blue points move continuously in the \birthdeath{} plane. 
The distance between any pair of red-blue points change continuously. The bottleneck pair between $\perone{\bp(s)}$ and $\pertwo{\nbp(t)}$ typically remains the same till certain \emph{critical values} of the parameters $(s,t)$. 

\begin{wrapfigure}{r}{0.35\textwidth}
\begin{tabular}{cc}
\includegraphics[height=2.7cm]{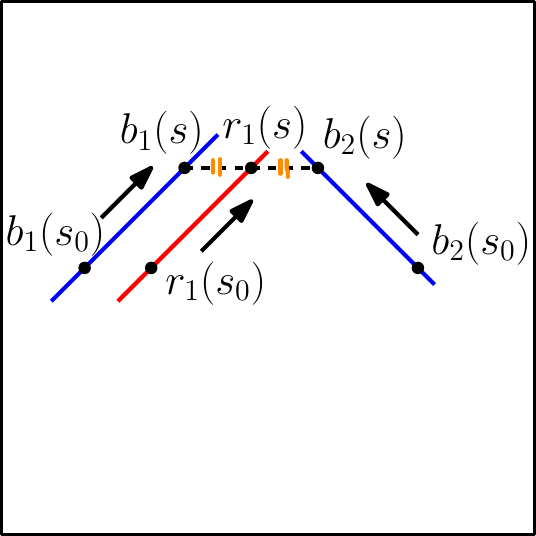} & \includegraphics[height=2.7cm]{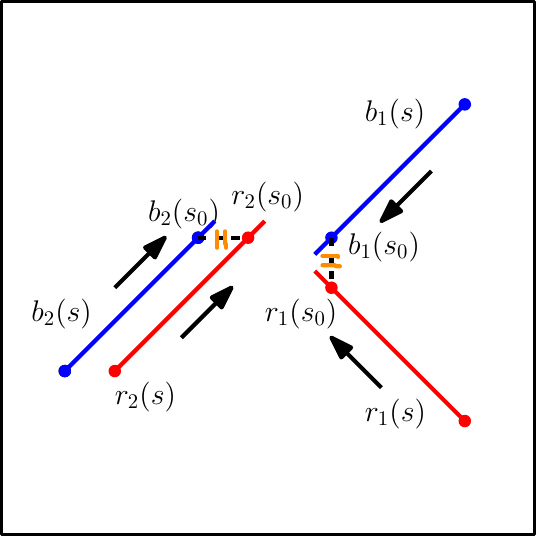}\\
(case-1) & (case-2)
\end{tabular}
\end{wrapfigure}
\vspace*{0.08in}\noindent{{\bf Characterizing critical $(s,t)$ values.}}~
Given $(s,t)$, consider the optimal bottleneck matching $C^*(s,t): \perone{s} \times \pertwo{t}$. 
For any corresponding pair $(r(s),b(t)) \in C^*(s,t)$, $d_\infty(r(s), b(t)) \le d_\infty(r^*(s), b^*(t))$. 
Suppose $r^*(s)=r_1(s)$ and $b^*(t) = b_1(t)$. 
As $(s,t)$ varies in $\recR$, the bottleneck pair $(r^*(s), b^*(t))$ may change only when: 
\begin{itemize}\denselist
\item {\it (case-1)}: ($r_1(s), b_1(t)$)  ceases to be a matched pair in the optimal matching $C^*(s,t)$; or
\item {\it (case-2)}: ($r_1(s), b_1(t)$) is still in $C^*$, but another matched pair $(r_2(s), b_2(t))$ becomes the bottleneck pair. 
\end{itemize}

At the time $(s_0, t_0)$ that either cases above happens, it is necessary that there are two red-blue pairs, one of which being $(r_1, b_1)$, and denoting the other one by $(r_2, b_2)$,  such that $d_\infty( r_1(s_0), b_1(t_0)) = d_\infty(r_2(s_0), b_2(t_0))$. 
(For case-1, we have that either $r_2 = r_1$ or $b_2 = b_1$.)
Hence all critical $(s,t)$ values are included in those $(s,t)$ values for which two red-blue pairs of persistence-points acquire equal distance in the \birthdeath{} plane. 
Let 
$$\pcritical_{(r_1,b_1),(r_2,b_2)}:= \{ (s,t) \mid d_\infty( r_1(s), b_1(t)) = d_\infty(r_2(s), b_2(t)) \}$$ 
denote the set of \emph{potential critical (s,t)-values generated by $(r_1,b_1)$ and $(r_2,b_2)$}. 
To describe $\pcritical_{(r_1,b_1),(r_2,b_2)}$, we first consider, for a fixed pair of red-blue points $(r, b)$, the distance function 
$\Drb_{r,b}: [0, L_s] \times [0, L_t] \to \reals$ defined as the distance between this pair of red and blue points in the \birthdeath{} plane, that is, $\Drb_{r,b}(s,t) := d_\infty(r(s), b(t))$ for any $(s,t)\in \recR$. 

In particular, recall that by Theorem \ref{thm:trajectories}, $r: [0, L_s] \to \mathbb{R}^2$ (resp. $b: [0, L_t] \to \mathbb{R}^2$) is continuous and piecewise-linear with $O(m)$ segments. 
In other words, the range $[0, L_s]$ (resp. $[0, L_t]$) can be decomposed to $O(m)$ intervals such that within each interval, $r$ moves (resp. $b$  moves) along a line in the \birthdeath{} plane with fixed speed. 
Hence combining Propositions \ref{prop:tracksaddle} and \ref{prop:trackmax}, we have the following:
\begin{proposition}\label{prop:oneredbluepair}
The s-t domain $\recR$ can be decomposed into an $O(m)\times O(m)$ grid such that, within each of the $O(m^2)$ grid cell, $\Drb_{r,b}$ is piecewise-linear with $O(1)$ linear pieces, and the partial derivative of each piece w.r.t. $s$ or w.r.t. $t$ is either `1', `-1', or `0'. 
\end{proposition}
\begin{proof}
\label{appendix:prop:oneredbluepair}
Let $\mathcal{I}_s$ (resp. $\mathcal{I}_t$) denote the decomposition of $[0, L_s]$ (resp. $[0, L_t]$) into $O(m)$ intervals within each of which the red (persistence) point $r\in \perone{\bp}$ (resp. the blue persistence point $b \in \pertwo{\nbp}$) moves along a line in the \birthdeath{} plane. 
In fact, by Propositions \ref{prop:tracksaddle} and \ref{prop:trackmax}, we also have that the birth-coordinate $r.x$ for the red point $r$ either increases or decreases at the unit speed (w.r.t. the parameter $s$), and the death-coordinate $r.y$ of $r$ either increases or decreases at the unit speed, or is stationary. 
Similar statements hold for the blue point $t$. 
Since $\Drb_{r,b}(s,t) = d_\infty (r(s), b(t)) = \max \{ | r.x(s) - b.x(t)|, | r.y(s) - b.y(t)| \}$, it follows that 
for a fixed interval $I_1 \in \mathcal{I}_s$ and $I_2 \in \mathcal{I}_t$, $\Drb_{r,b}: I_1 \times I_2 \to \mathbb{R}$ is piecewise-linear function with $O(1)$ linear pieces, where  the partial derivative of each piece w.r.t. $s$ or to $t$ is either `1', `-1', or `0'. 
\end{proof}

Given two pairs of red-blue pairs $(r_1, b_1)$ and $(r_2, b_2)$, the set $\pcritical_{(r_1,b_1),(r_2,b_2)}$ of potential critical (s,t) values generated by them corresponds to the intersection of the graph of $\Drb_{r_1,b_1}$ and that of $\Drb_{r_2,b_2}$. 
By overlaying the two $O(m)\times O(m)$ grids corresponding to $\Drb_{r_1,b_1}$ and $\Drb_{r_2,b_2}$ as specified by Proposition \ref{prop:oneredbluepair}, we obtain another grid of size $O(m)\times O(m)$ and within each cell, the intersection of the graphs of $\Drb_{r_1,b_1}$ and $\Drb_{r_2,b_2}$ has $O(1)$ complexity. 
Hence, we have: 
\begin{cor}\label{cor:potentialset}
The set $\pcritical_{(r_1,b_1),(r_2,b_2)} \subseteq \recR$ consists of a set of polygonal curves in the s-t domain $\recR$ with $O(m^2)$ total complexity. 
\end{cor}

Consider the arrangement $Arr(\recR)$ of the set of curves in $\mathcal{\pcritical} = \{ \pcritical_{(r_1,b_1), (r_2,b_2)} \mid r_1, r_2 \in \perone{\bp}, b_1, b_2 \in \pertwo{\nbp} \}$. 
Since there are altogether $O(m^4) \times O(m^2) = O(m^6)$ segments in $\mathcal{\pcritical}$, we have that the arrangement $Arr(\recR)$ has $O(m^{12})$ complexity; that is, there are $O(m^{12})$ number of vertices, edges and polygonal cells. 
However, this arrangement $Arr(\recR)$ is more refined than necessary. 
Specifically, within a single cell $c \in Arr(\recR)$, the \emph{entire} bottleneck matching $C^*$ does not change. 
By a more sophisticated argument, we can improve the complexity as follows: 
\begin{proposition}\label{prop:arrangement}
There is a planar decomposition $\dcomp(\recR)$ of the s-t domain $\recR$ with $O(m^8)$ number of vertices, edges and polygonal cells such that as (s,t) varies within in each cell $c \in \dcomp(\recR)$, the pair of red-blue persistence points that generates the bottleneck pair $(r^*, b^*)$ remains the same.

Furthermore, the decomposition $\dcomp(\recR)$, as well as the bottleneck pair $(r^*,b^*)$ associated to each cell, can be computed in $O(m^{9.5}\log m)$ time. 
\end{proposition}
\begin{proof}
\label{appendix:prop:arrangement}
First, consider the decomposition of $\recR$ into maximal cells within each of which the bottleneck pair does not change its identity. We refer to each such cell as a \emph{fixed-bottleneck-pair cell}. 
Consider such a cell $c$ and assume that within this cell $c$ the bottleneck pair is $(r^*, b^*) = (r_1,b_1)$. The boundary of $c$ is a polygonal curve $\gamma$, each linear segment of which corresponds to potential critical (s,t)-values where the red-blue pair $(r_1,b_1)$ has equal distance with some other red-blue pair, say $(r_2, b_2)$. 
Each vertex, say $v$ in this boundary curve $\gamma$ is where two segments meet, say one corresponding to $(r_1,b_1)$ and $(r_2, b_2)$, and the other corresponding to $(r_1,b_1)$ and $(r_3,b_3)$. 

The vertices in $\gamma$ are of two types: 
(Type-1): $(r_2, b_2)=(r_3, b_3)$ where $v$ is also a vertex in a polygonal curve from $X_{(r_1,b_1),(r_2,b_2)}$; 
(Type-2): remaining case where $v = (s_0, t_0)$ represents the moment 
the red-blue pair $(r_1,b_1)$ has distance equal to that of the two other red-blue pairs: $(r_2,b_2)$ and $(r_3,b_3)$. 


Type-1 vertices are vertices from the same polygonal curve of where $(r_1,b_1)$ and $(r_2, b_2)$ are at equal distances. 
Type-2 vertices are where this curve meets another curve representing the (s,t)-values where $(r_1,b_1)$ and $(r_3,b_3)$ are at equal distances. 
Hence Type-2 vertices represent the places where \emph{three} fixed-bottleneck-pair cells meet. 

By Corollary \ref{cor:potentialset}, there are $O(m^4)\times O(m^2) = O(m^6)$ number of Type-1 vertices. 
We now show that the number of Type-2 vertices is $O(m^8)$ and we can compute all Type-2 vertices in $O(m^{9.5}\log m)$ time.  

Note that each Type-2 vertex is induced by three red points ($r_1,r_2,r_3$) and three blue points ($b_1,b_2,b_3$). 
First, enumerate all $O(m^6)$ possible triples of red-blue pairs. 
For each triple $(r_1,b_1)$, $(r_2, b_2)$ and $(r_3,b_3)$, consider the graphs of functions $\Drb_{r_1,b_1}$, $\Drb_{r_2,b_2}$, and $\Drb_{r_3,b_3}$. 
The intersection of all three graphs are a super-set for Type-2 vertices generated by $(r_1,b_1)$, $(r_2, b_2)$ and $(r_3,b_3)$. It follows from Proposition \ref{prop:oneredbluepair} that there are $O(m^2)$ intersection points of the three graphs -- Specifically, we overlay the three $O(m) \times O(m)$ grids as specified by Proposition \ref{prop:oneredbluepair}, and within each cell of the resulting grid which is still of size $O(m) \times O(m)$, each function $\Drb_{r_i,b_i}$ has $O(1)$ complexity, and thus they produce $O(1)$ intersection points. 
For each intersection point, we spend $O(m^{1.5}\log m)$ time using the modified algorithm of \cite{EKI01} to compute its bottleneck matching, and check whether this is a valid Type-2 vertex or not. 
Altogether, since there are $O(m^6)$ triples we need to check, there can be $O(m^8)$ Type-2 vertices, and they can be identified in $O(m^{9.5}\log m)$ time. Let $\Sigma$ denote the resulting set of Type-2 vertices. 

With $\Sigma$ computed, we next construct the decomposition $\dcomp(\recR)$ of $\recR$ into fixed-bottleneck-pair cells. 
To do this, we simply scan all vertices in $\Sigma$ from left to right. For each vertex $v\in \Sigma$ corresponding to the three red-blue pairs $(r_1,b_1)$, $(r_2, b_2)$ and $(r_3,b_3)$, we know that locally, there are three branches from $v$: one from $\pcritical_{r_1,b_1,r_2,b_2}$, one from $\pcritical_{r_1,b_1,r_3,b_3}$ and one from $\pcritical_{r_2,b_2,r_3,b_3}$. 
We simply trace each such curve till we meet another vertex from $\Sigma$. 
Now consider the graph whose nodes are Type-2 vertices, and arcs are polygonal curves connecting them. 
We can use any graph traversal strategy (such as BFS) to traverse all arcs and thus connecting nodes. 
The total time is 
$$
O(\text{Time to traverse the graph}) + O(\text{Time to trace out all arcs}). $$ 
Since this graph is planar with $O(m^8)$ vertices, there are $O(m^8)$ arcs as well. Hence $O$(Time to traverse the graph) $= O(m^8)$. 
The time to trace an arc from one Type-2 vertex to the other is proportional to the complexity of this polygonal curve.  
We charge this time to the number of interior Type-1 vertices in this arc, as well as the two boundary Type-2 vertices of this arc. By Corollary \ref{cor:potentialset}, the polygonal curves from $\mathcal{\pcritical}$ has $O(m^2) \times O(m^4) = O(m^6)$ total complexity. 
Hence the total time to trace out all arcs is also bounded by $O(m^8)$.  Putting everything together, we have that, once the Type-2 vertices are computed, we can construct $\dcomp(\recR)$ in time $O(m^8)$. This completes the proof. 
\end{proof}

Our goal is to compute the bottleneck distance function $\onePD: \recR \to \reals$ introduced at the beginning of this subsection where
 $\onePD(s, t) \mapsto d_B (\perone{\bp(s)}, \pertwo{\nbp(t)}) = d_\infty (r^*(s), b^*(t)), $ 
so as to further compute \spdist{} distance using Eqn (\ref{eqn:PDthree}).
To do this, we need to further refine the decomposition $\dcomp(\recR)$ from Proposition \ref{prop:arrangement} to another decomposition $\augdcomp(\recR)$ as described below so that within each cell, the bottleneck distance function $\onePD_{\ebase_s,\ebase_t}$ can be described by a single linear function. 

\begin{theorem}\label{thm:bottleneckdistfunc}
For a fixed pair of edges $\ebase_s\in \graphone$ and $\ebase_t \in \graphtwo$, there is a planar polygonal decomposition $\augdcomp(\recR)$ of the s-t domain $\recR$ of  $O(m^{10})$ complexity such that within each cell, the bottleneck distance function $\onePD_{\ebase_s,\ebase_t}$ is linear. 
Furthermore, one can compute this decomposition $\augdcomp(\recR)$ as well as the function $\onePD_{\ebase_s,\ebase_t}$ in $O(m^{10}\log m)$ time. 
%
\end{theorem}
\begin{proof}
\label{appendix:thm:bottleneckdistfunc}
By Proposition \ref{prop:arrangement}, given any cell $c \in \dcomp(\recR)$, the bottleneck pair $(r^*, b^*)$ remains the same. In other words, let $r_c = r^*$ and $b_c = b^*$ for any $(s,t) \in c$. 
We have $\onePD(s,t) = d_\infty(r_c(s), b_c(t)) (= \Drb_{r_c,b_c}(s,t) )$ for $(s,t)\in c$. 
Let $A_{r_c,b_c}(\recR)$ be the decomposition of $\recR$ such that within each cell of $A_{r_c, b_c}$, the function $\Drb_{r_c,b_c}$ is a linear function. By Proposition \ref{prop:oneredbluepair}, $A_{r_c,b_c}$ consists of $O(m^2)$ cells, edges and vertices.  
Hence we can further decompose (refine) the cell $c$ to be the intersection of $c$ with $A_{r_c, b_c}$. 
We perform this refinement for each cell $c \in \dcomp(\recR)$, and denote the resulting decomposition as $\augdcomp(\recR)$. 
By construction, the bottleneck distance $\onePD$ within each cell of $\augdcomp(\recR)$ is a single linear function. 

Next, we bound the complexity of $\augdcomp(\recR)$. First, note that the number of newly added vertices within the interior of a cell $c\in \dcomp(\recR)$ is bounded from above by $O(m^2)$,  since each such vertex is a vertex from $A_{r_c,b_c}$. 
While there can be $O(m^8)$ number of cells in $\dcomp(\recR)$, there can only be $O(m^2)$ choices of bottleneck pairs ($r^*,b^*$)s. Hence the total number of vertices in the interior cells in $\dcomp(\recR)$ is $O(m^4)$. 
 
What remains is to bound the number of vertices along edges of $\dcomp(\recR)$. 
To this end, notice that each edge $e \in \dcomp(\recR)$ has two incident cells $c_1$ and $c_2$. Any newly added vertex in $e$ must be either an intersection between $e$ with some edge in $A_{r_{c_1}, b_{c_1}}$, or with some edge in $A_{r_{c_2}, b_{c_2}}$. Hence the total number of such vertices on $e$ is $O(m^2)$. 
Since there are $O(m^8)$ edges in $\dcomp(\recR)$, the total number of newly added vertices is at most $O(m^{10})$. Thus the complexity of $\augdcomp(\recR)$ is $O(m^{10})$. 

Finally, the refined decomposition $\augdcomp(\recR)$ can be computed in $O(m^{10}\log m)$ time. Specifically, first, it takes $O(m^{9.5}\log m)$ time to compute $\dcomp(\recR)$ by Proposition \ref{prop:arrangement}. Next, for each cell $c$ with $k$ number of boundary edges, it takes $O((k+m^2)\log m)$ time to compute the intersection $c \cap A_{r_c,b_c}$. Summing over all cells in $\dcomp(\recR)$ gives the claimed time complexity. 
\end{proof}

\subsubsection{Final algorithm and analysis.}
\label{subsec:altogether}

\begin{wrapfigure}{r}{0.4\textwidth}
\begin{tabular}{cc}
\includegraphics[height=3.3cm]{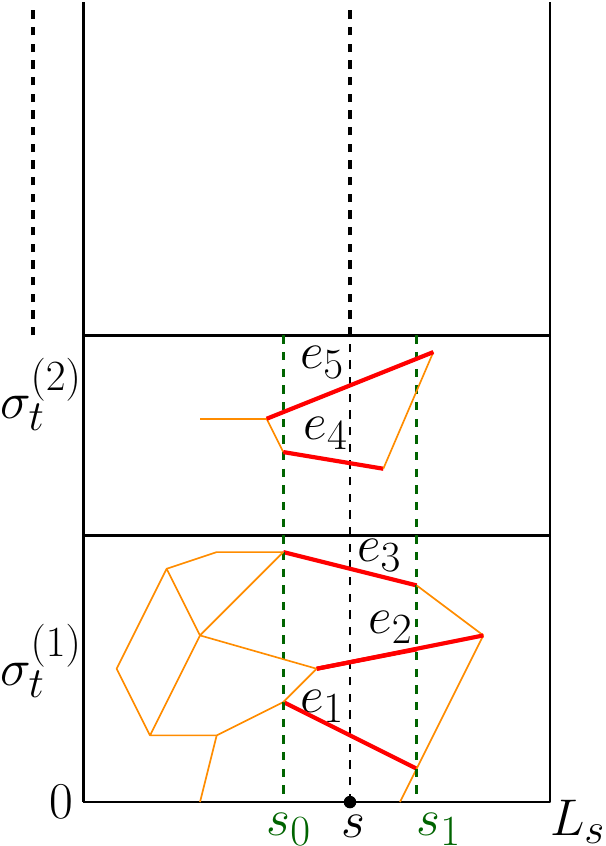} & \includegraphics[height=2.7cm]{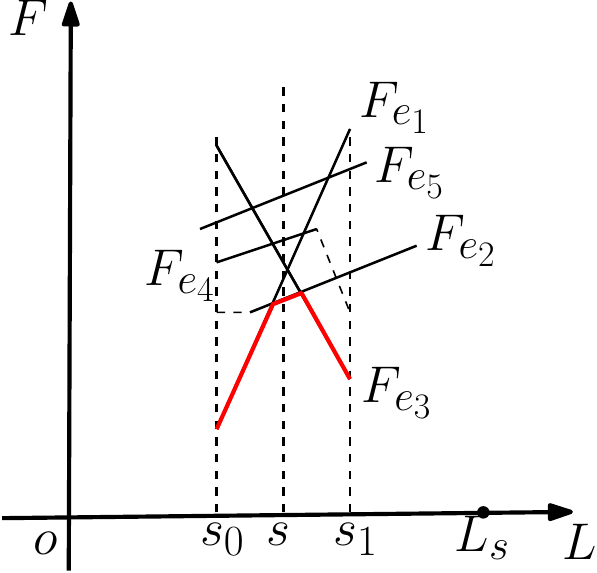}\\
(a) & (b) 
\end{tabular}
\caption{(a) s-t domains for $\ebase_s\in \Eone$ and edges $\ebase^{(j)}_t \in \Etwo$. (b) $\Lenv(s)$ is the lowest value along any $\onePD_{e_\ell}$. 
\label{fig:lowerenv}}
\end{wrapfigure}
We now aim to compute $\vec{d}_H(\setone,\settwo)$ using Eqn (\ref{eqn:PDthree}). 
First, for a fixed edge $\ebase_s \in \graphone$, consider the following \emph{lower-envelop function} 
\begin{equation}\label{eqn:envelop}
\Lenv: [0, L_s] \to \reals~~\text{where}~~\Lenv(s) \mapsto \min_{\ebase_t \in \graphtwo}\min_{t\in [0, L_t]} \onePD(s,t), 
\end{equation}
where recall $L_s$ and $L_t$ denote the length of edge $\ebase_s$ and $\ebase_t$ respectively. 
The reason behind the name ``lower-envelop function" will become clear shortly. 

Now for each $\ebase_t \in \graphtwo$, consider the polygonal decomposition $\augdcomp(\recR)$ as described in Theorem \ref{thm:bottleneckdistfunc}. 
Since within each cell the bottleneck distance function $\onePD$ is a linear piece, we know that for any $s$, the extreme of $\onePD(s,t)$ for all possible $t \in [0, L_t]$ must come from some edge in $\augdcomp(\recR)$. 
In other words, to compute the function $\min_{t \in [0, L_t]} \onePD(s,t)$ at any $s\in [0, L_s]$, we only need to inspect the function $\onePD$ restricted to edges in the refined decomposition $\augdcomp(\recR_{\ebase_s,\ebase_t})$ for the s-t domain $\recR_{\ebase_s,\ebase_t} = [0, L_s]\times [0, L_t]$. 
Take any edge $e$ of $\augdcomp(\recR_{\ebase_s,\ebase_t})$, define $\pi_e: [0, L_s] \to [0, L_t]$ such that $(s, \pi_e(s)) \in e$. Now denote by the function $\onePD_e: [0, L_s] \to \reals$ as the projection of $\onePD$ onto the first parameter $[0, L_s]$; that is, $\onePD_e(s) := \onePD(s, \pi_e(s))$. 
Let $E_{\ebase_s}:= \{ e \in \augdcomp(\recR_{\ebase_s,\ebase_t}) \mid \ebase_t \in \graphtwo \}$ be the union of edges from the refined decompositions of the s-t domain formed by $\ebase_s$ and any edge $\ebase_t$ from $\graphtwo$. 
It is easy to see that (see Figure \ref{fig:lowerenv}): 
$$
\Lenv(s) = \min_{e\in E_{\ebase_s}} \onePD_e(s); ~~\text{that is,}~~\Lenv~\text{is the lower-envelop of linear functions}~\onePD_e~\text{for all }e\in E_{\ebase_s}. 
$$
There are $O(m)$ edges in $\graphtwo$, thus by Theorem \ref{thm:bottleneckdistfunc} we have $|E_{\ebase_s}| = O(m^{11})$. 
The lower envelop $\Lenv$ of $|E_{\ebase_e}|$ number of linear functions (linear segments), is a piecewise-linear function with $O(|E_{\ebase_s}|=O(m^{11})$ complexity and can be computed in $O(|E_{\ebase_s}|\log |E_{\ebase_s}|) = O(m^{11}\log m)$ time. 
Finally, from Eqn (\ref{eqn:PDthree}), $\vec{d}_H(\setone, \settwo) = \max_{\ebase_s\in \graphone} \max_{s\in [0,L_s]} \Lenv(s)$. Since there are $O(m)$ choices for $\ebase_s$, we conclude with the following main result. 
\begin{theorem}
Given two metric graphs $(\graphone, d_{G_1})$ and $(\graphtwo, d_{G_2})$ with $n$ total vertices and $m$ total edges, we can compute the \spdist{} distance $\dsp(\graphone,\graphtwo)$ between them in $O(m^{12}\log n)$ time. 
\end{theorem}

We remark that if both input graphs are metric trees, then we can compute their \spdist{} distance more efficiently in $O(n^8\log n)$ time. 

\section{Preliminary Experiments}
\label{appendix:sec:exp}

We show two sets of preliminary experimental results. 
The first experiment aims to demonstrate the stability of the proposed \spdist{} distance, by showing that the \spdist{} distance between a graph and a noisy sample of it remains stable w.r.t. the noise added. 
In the second experiment, we apply our \spdist{} distance to compare a set of surface models, using simply the 1-skeleton of their mesh models, and show that this distance is robust against non-rigid but near-isometric deformations (such as different poses between humans, or between wolfs and horses), while still differentiating different models. In both experiments, to improve the efficiency, we only compute the persistence diagrams to a subset of graph nodes of input graphs, and obtain an even coarser version of the discrete \spdist{} distance for input graphs. 

We also point out that in our experiments, we compute the 0-th \emph{zigzag persistence diagram} for each basepoint. However, we observe little difference in results if only the 0-th standard persistence diagrams are used. 

\begin{figure}[tbhp]
\begin{center}
\begin{tabular}{ccc}
\includegraphics[width=5.5cm]{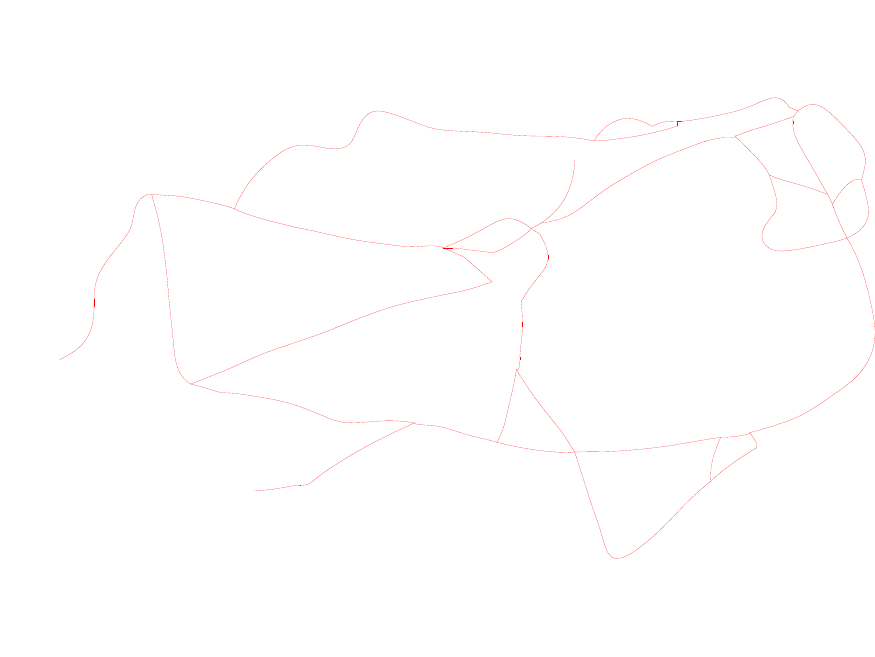} & \hspace*{0.1in} & \includegraphics[width=5.5cm]{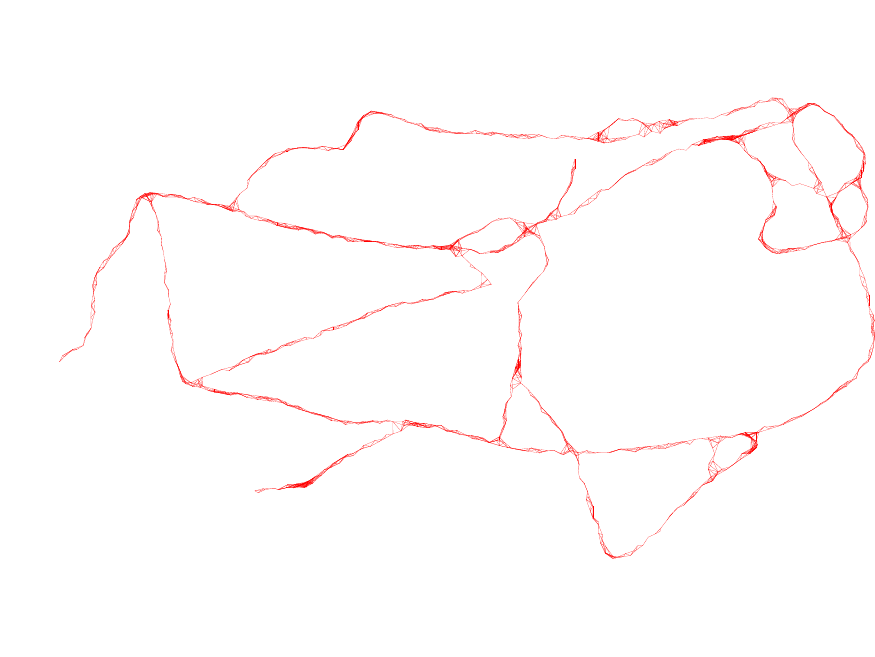}\\
(a) & & (b)\\
\includegraphics[width=5.5cm]{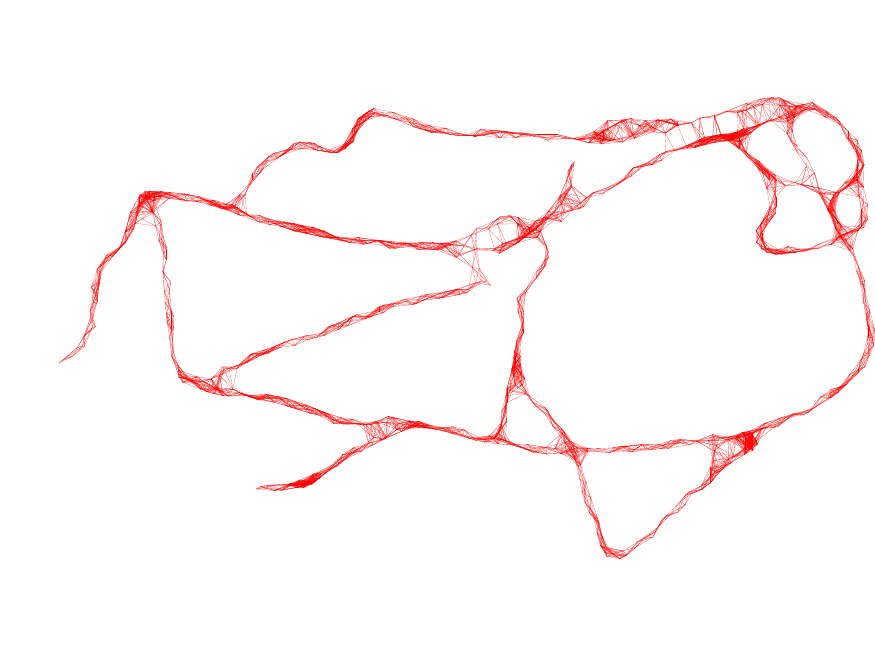}& \hspace*{0.1in} & \includegraphics[width=5.5cm]{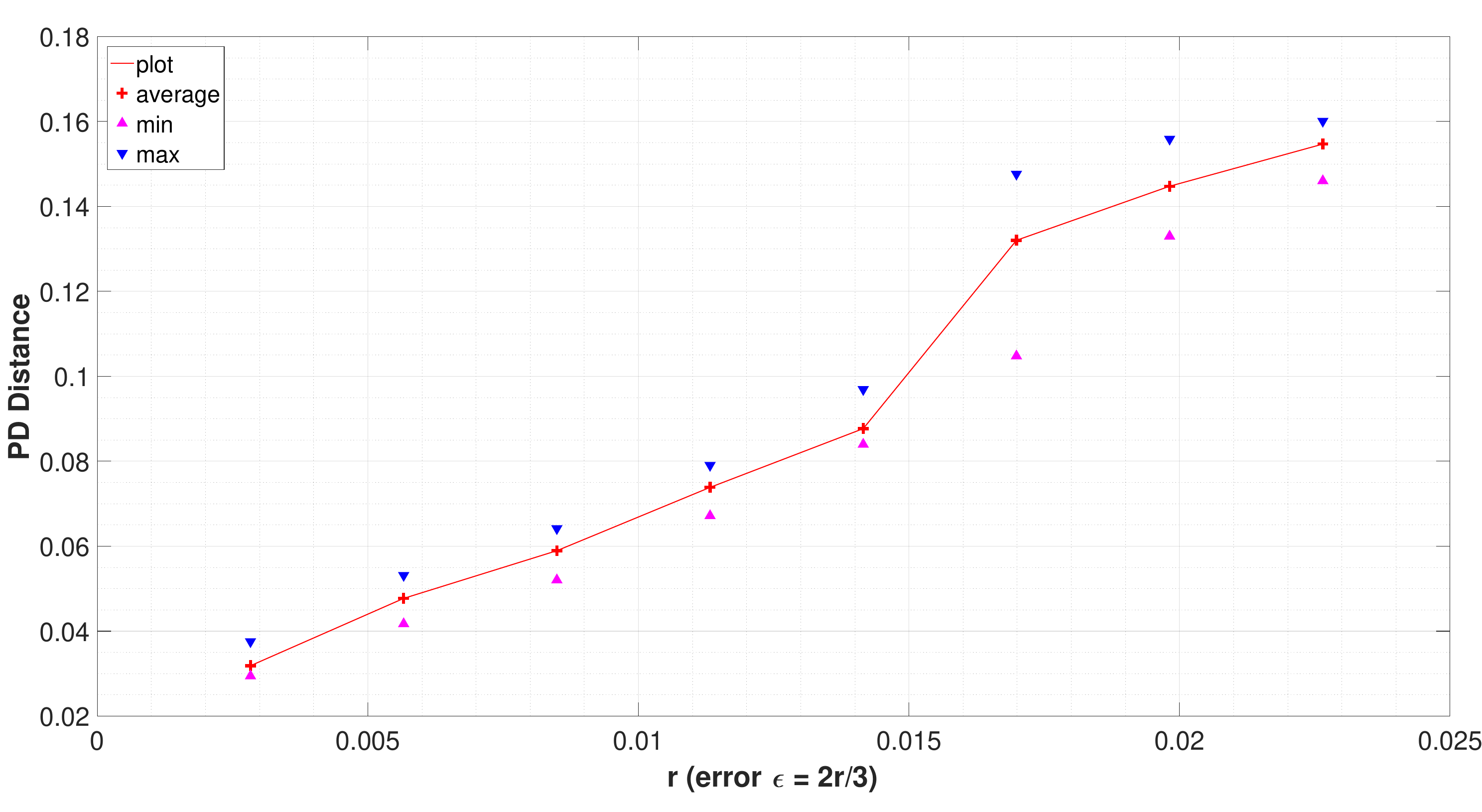}\\
(c) & & (d)
\end{tabular}
\end{center}
\caption{(a) Hidden graph (Athens road map) $G$. (b) and (c) noisy 1-skeleton of Rips complex $\mathcal{R}^r$ for $r = 0.011$ and $r = 0.02$ respectively. (d) The growth of the \spdist{} distance $\dsp(\mathcal{R}^r_1, G)$ w.r.t. the parameter $r$ in the Rips complex $\mathcal{R}^r$. Note that the noise level is $\eps = \frac{2r}{3}$. The vertical range (two triangle-points) shows the max and min $\dsp$ values for 10 different re-samples (with noise) (i.e, for each  noise level, we take 10 sets of samples) -- the middle curve is the average $\dsp$ values of these 10 sets.  
\label{fig:athens}}
\end{figure}
\paragraph{Experiment 1.}
The first experiment aims to demonstrate the stability of the proposed \spdist{} distance. Specifically, we consider a set of noisy points $P$ sampled from a hidden graph $G=(V,E)$, and compute the the Rips complex $\mathcal{R}^r(P)$ of $P$ as an approximation of the hidden graph $G$. 
The hidden graph $G$ taken in this case is a part of the Athens road network. We obtain a noisy sample $P_\eps$ by uniformly sampling points from each edge with distance $\eps$ and perturbing sample points within the circular region of radius $\frac{\eps}{4}$.
We then build a Rips complex $\mathcal{R}^{r(\eps)}(P_\eps)$ with parameter $r(\eps) = \frac{3\eps}{2}$. 
We then treat the 1-skeleton $\mathcal{R}^{r(\eps)}_1$ of $\mathcal{R}^{r(\eps)}(P_\eps)$ as a metric graph, and this metric graph $\mathcal{R}^{r(\eps)}_1$ offers a noisy approximation of the hidden graph $G$. 
Examples of $G$ and $\mathcal{R}^{r(\eps)}_1$s are shown in Figure \ref{fig:athens} (a), (b) and (c). 

To speedup the computation, we compute only the persistence diagrams at a set of \emph{$\delta$-sparse} subsamples $Q \subset V$ and $Q' \subset P_\eps$, and only use points in $Q$ and $Q'$ as basepoints. Specifically, $Q$ is obtained by the following randomized procedure: Take a random permutation of $V$. Process each node $v_i \in V$ in this order. We add $v_i$ into $Q$ \emph{only} if its distance to current points in $Q$ is larger than $\delta$. 
The point set $Q'$ is obtained from $P_\eps$ in a similar manner. 
We use a random order so as to further demonstrate the robustness against different discretization. 
By Theorem \ref{thm:stability}, one can show that this incurs at most $12 \delta$ error in the estimation of $\dsp(\mathcal{R}^{r(\eps)}_1, G)$. 
In our experiments, the size of the subsampled sets $Q$ and $Q'$ are usually between 150 and 200 points. The time required for computing the discrete \spdist{} distance using $Q$ and $Q'$ as basepoints, is observed to be from $20 \sim 30$ seconds.

In Figure \ref{fig:athens} (d), we show the growth of the \spdist{} distance $\dsp(\mathcal{R}^{r(\eps)}_1, G)$ with respect to the change of the noise level $\eps$; recall that the parameter $r(\eps)= \frac{3\eps}{2}$. We note that $\dsp(\mathcal{R}_1, G)$ grows roughly proportionally to the noise level, demonstrating its stability. We note that there is a small jump of $\dsp(\mathcal{R}_1, G)$ from $r(\eps)=0.017$ to $r(\eps)=0.02$. This is because when $r(\eps)$ increases, small loops (1st homology features) get created in the top-right part of the graph (Figure \ref{fig:athens} (c)). This shows that our \spdist{} distance captures such small topological changes.

\begin{figure}[tbhp]
\begin{center}
\includegraphics[width=12cm]{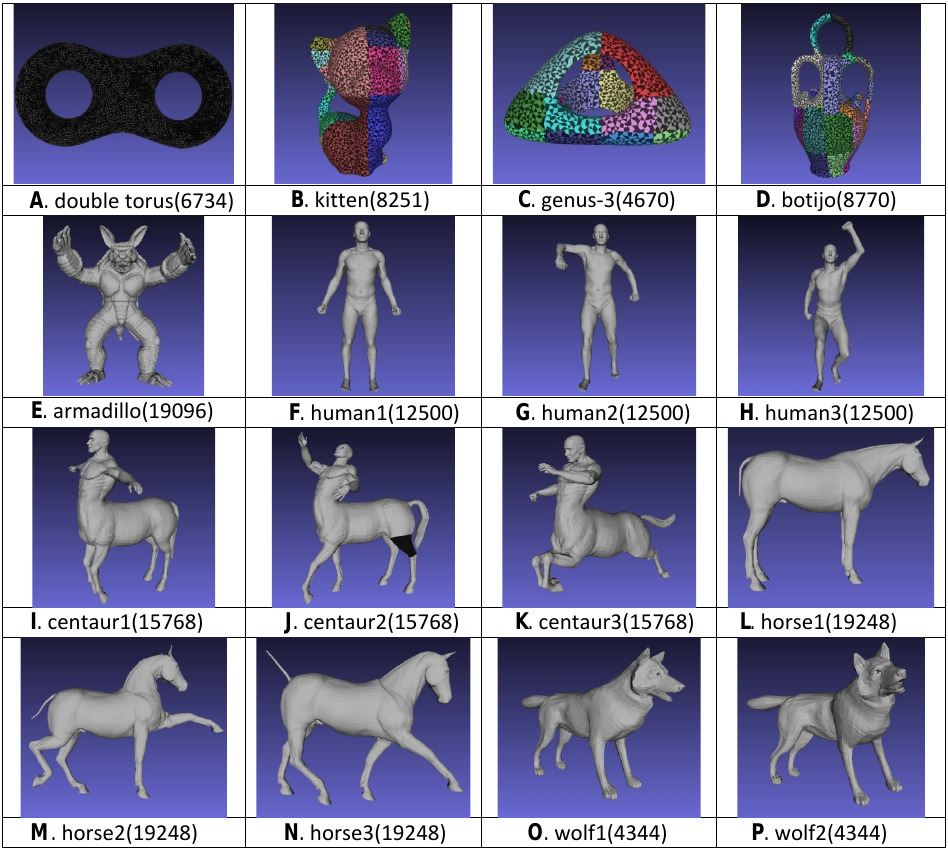}
\caption{Models used for comparison.
\label{fig:surfacemodels}}
\end{center}
\end{figure}
\begin{figure}[tbhp]
\begin{center}
\includegraphics[width=15cm]{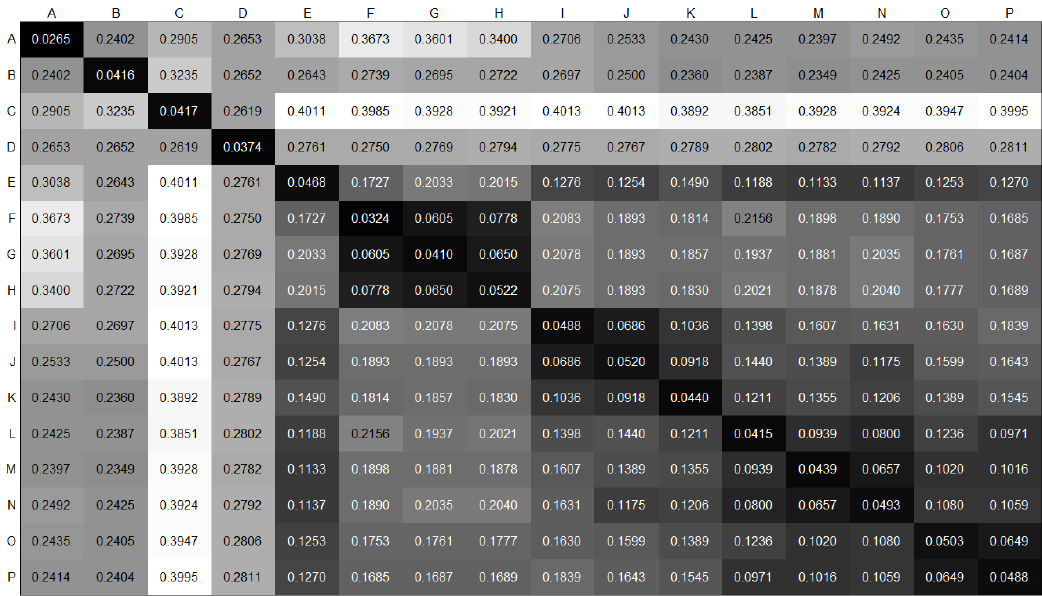}
\caption{Pairwise \spdist{} distances between models.
\label{fig:surfacematrix}}
\end{center}
\end{figure}
  
\paragraph{Experiment 2.} 
In the second experiment, we apply our \spdist{} distance to compare surface meshes of different geometric models, some of which are different poses of the same object. The set of surface models are shown in Figure \ref{fig:surfacemodels}. For each surface model, we take the 1-skeleton $K_i = (V_i, E_i)$ of its surface mesh as input. 
As in the first experiment, we also compute only the persistence diagrams at a set of $\delta$-sparse subsamples $Q_i \subset V_i$ from input surface mesh constructed by a randomized decimation procedure. Again by Theorem \ref{thm:stability}, one can show that this incurs at most $12 \delta$ error in the estimation of $\dsp(K_i, K_j)$. 

Figure \ref{fig:surfacematrix} shows the matrix of pairwise \spdist{} distance between all pairs of models. Because the subsamples are generated by a randomized procedure, the resulting \spdist{} distance for the same two graphs may be non-zero (as the set of basepoints chosen may be different). Nevertheless, note that distance values at the diagonal are usually small, implying that \spdist{} is stable against different discretization of the same graph. 

From the matrix in Figure \ref{fig:surfacematrix}, we can see that models from the same group (such as human1, human2 and human3) have very small \spdist {} distances among them (darker colors for smaller values). Furthermore, models from similar groups (such as between wolves and horses) have \spdist{} distances smaller than those between dissimilar groups (such as between wolves and double-torus). This demonstrates that our \spdist{} distance is a reasonable measure for differentiating surface models. 

The number of vertices of an input mesh for each model is shown in brackets in Figure \ref{fig:surfacemodels} after the model name. 
The size of the subsample of a graph is usually kept between 200 and 300. 
The time for computing the \spdist{} distance is typically less than 10 seconds. For the exceptional case 
involving two armadillos where the input graphs have large sizes, the running time is around 20 seconds. 


We also remark that it is possible to take simply the 1-skeleton of the Rips complex constructed from the point samples $V_i$ of a surface mesh instead of using the surface mesh itself. We expect to obtain similar results though the complex size will most likely be larger.

\section{Conclusions and Future Directions}
\label{sec:conclusion}

In this paper, we proposed a new way to measure distance between metric graphs, called the \spdist{} distance. This distance is developed based on a topological idea, and provides a new angle to the metric graph comparison problem. 
The proposed \spdist{} distance is stable with respect to metric distortion, and align the underlying space of input graphs (instead of just graph nodes). 
Despite considering all points in input graphs, we show that a polynomial time algorithm exists for computing the \spdist{} distance. 
We have implemented the discrete version of our \spdist{} distance for graphs, which is available at \cite{GM2014}.

The time complexity for computing the (continuous) \spdist{} distance is
high. A worthwhile endeavor will be to bring it down with more accurate 
analysis. In particular, the geodesic distance function (to a basepoint) in the graph has many special properties, some of which we already leverage. It will be interesting to see whether we can further leverage these properties to reduce the bound on the decomposition $\augdcomp(\Omega)$ as used in Theorem \ref{thm:bottleneckdistfunc}. 
Developing efficient approximation algorithms for computing the \spdist{} distance is also an interesting question. Also, the special case of metric trees is worthwhile to investigate. Notice that even discrete tree matching is still a hard problem for unlabeled trees, i.e, when no correspondences between tree nodes are given. 

\paragraph{Acknowledgment.} We thank anonymous reviewers for very helpful comments, including the suggestion that $d_B(\perone \bp, \pertwo \nbp)$ can be computed directly using the algorithm of \cite{EKI01}, which simplifies our original approach based on modifying the algorithm of \cite{EKI01}. 
This work is partially supported by NSF under grants CCF-0747082, CCF-1064416, CCF-1319406, CCF1318595.

\bibliographystyle{abbrv}
\bibliography{ref}

\begin{thebibliography}{10}

\bibitem{GM2014}
{GraphComp Software}, 2014.
\newblock Project URL:
  \url{http://web.cse.ohio-state.edu/~tamaldey/paper/graph-match/GraphComp-software/}.

\bibitem{ACC12}
M.~Aanjaneya, F.~Chazal, D.~Chen, M.~Glisse, L.~Guibas, and D.~Morozov.
\newblock Metric graph reconstruction from noisy data.
\newblock {\em Int. J. Comput. Geom. Appl.}, pages 305--325, 2012.

\bibitem{AFNSW15}
P.~K. Agarwal, K.~Fox, A.~Nath, A.~Sidiropoulos, and Y.~Wang.
\newblock {\em Computing the Gromov-Hausdorff Distance for Metric Trees}, pages
  529--540.
\newblock Springer Berlin Heidelberg, Berlin, Heidelberg, 2015.

\bibitem{AHU74}
A.~V. Aho, J.~E. Hopcroft, and J.~D. Ullman.
\newblock {\em {The Design and Analysis of Computer Algorithms}}.
\newblock Addison Wesley, 1974.

\bibitem{Babai16}
L.~Babai.
\newblock Graph isomorphism in quasipolynomial time [extended abstract].
\newblock In {\em Proc. 48th {ACM} {SIGACT} Sympos. Theory Comput. ({STOC})},
  pages 684--697, 2016.
\newblock Arxiv version available at: arXiv:1512.03547.

\bibitem{BGW14}
U.~Bauer, X.~Ge, and Y.~Wang.
\newblock Measuring distance bewteen {R}eeb graphs.
\newblock In {\em Proc. 30th SoCG}, pages 464--473, 2014.
\newblock Full version available at arXiv:1307.2839.

\bibitem{BBI01}
D.~Burago, Y.~Burago, and S.~Ivanov.
\newblock {\em A course in metric geometry}.
\newblock volume 33 of \emph{AMS Graduate Studies in Math}. American
  Mathematics Society, 2001.

\bibitem{CS14}
F.~Chazal and J.~Sun.
\newblock {Gromov-Hausdorff Approximation of Filament Structure Using Reeb-type
  Graph}.
\newblock In {\em Proc. 30th SoCG}, pages 491--500, 2014.

\bibitem{CEH07}
D.~Cohen-Steiner, H.~Edelsbrunner, and J.~Harer.
\newblock Stability of persistence diagrams.
\newblock {\em Discrete {\&} Computational Geometry}, 37(1):103--120, 2007.

\bibitem{CEH09}
D.~Cohen-Steiner, H.~Edelsbrunner, and J.~Harer.
\newblock Extending persistence using {Poincar{\'e} and Lefschetz} duality.
\newblock {\em Foundations of Computational Mathematics}, 9(1):79--103, 2009.

\bibitem{CEM06}
D.~Cohen-Steiner, H.~Edelsbrunner, and D.~Morozov.
\newblock Vines and vineyards by updating persistence in linear time.
\newblock In {\em Proc. 22nd SoCG}, pages 119--126, 2006.

\bibitem{CSS07}
T.~Cour, P.~Srinivasan, and J.~Shi.
\newblock {Balanced Graph Matching}.
\newblock In {\em Advances in Neural Information Processing Systems 19}, pages
  313--320. MIT Press, 2007.

\bibitem{DW07}
T.~K. Dey and R.~Wenger.
\newblock {Stability of critical points with interval persistence}.
\newblock {\em Discrete Comput. Geom.}, 38:479--512, 2007.

\bibitem{EH09}
H.~Edelsbrunner and J.~Harer.
\newblock {\em {Computational Topology: {An} Introduction}}.
\newblock Amer. Math. Soc., Providence, Rhode Island, 2009.

\bibitem{ELZ02}
H.~Edelsbrunner, D.~Letscher, and A.~Zomorodian.
\newblock Topological persistence and simplification.
\newblock {\em Discrete Comput. Geom.}, 28:511--533, 2002.

\bibitem{EKI01}
A.~Efrat, M.~Katz, and A.~Itai.
\newblock Geometry helps in bottleneck matching and related problems.
\newblock {\em Algorithmica}, 1:1--28, 2001.

\bibitem{FSV01}
P.~Foggia, C.~Sansone, and M.~Vento.
\newblock {A Performance Comparison of Five Algorithms for Graph Isomorphism}.
\newblock In {\em Proc. 10th Intl. Conf. Image Ana. Proc. (ICIAP)}, Italy,
  2001.

\bibitem{GXTL10}
X.~Gao, B.~Xiao, D.~Tao, and X.~Li.
\newblock A survey of graph edit distance.
\newblock {\em Pattern Anal. Appl.}, 13(1):113--129, Jan. 2010.

\bibitem{GSBW11}
X.~Ge, I.~Safa, M.~Belkin, and Y.~Wang.
\newblock Data skeletonization via {Reeb} graphs.
\newblock In {\em Proc. 25th NIPS}, pages 837--845, 2011.

\bibitem{GR96}
S.~Gold and A.~Rangarajan.
\newblock {A Graduated Assignment Algorithm for Graph Matching}.
\newblock In {\em IEEE Trans. on PAMI}, volume~18, pages 377--388, 1996.

\bibitem{Gromov99}
M.~Gromov.
\newblock {\em Metric structures for {Riemannian} and non-{Riemannian} spaces}.
\newblock volume 152 of \emph{Progress in Mathematics}. Birkh\"{a}user Boston
  Inc., 1999.

\bibitem{HW74}
J.~E. Hopcroft and J.~K. Wong.
\newblock {Linear Time Algorithm for Isomorphism of Planar Graphs (Preliminary
  Report)}.
\newblock In {\em Proc. of the ACM STOC}, STOC '74, pages 172--184, New York,
  NY, USA, 1974. ACM.

\bibitem{HRG13}
N.~Hu, R.~Rustamov, and L.~Guibas.
\newblock {Graph Matching with Anchor Nodes: A Learning Approach}.
\newblock In {\em IEEE Conference on CVPR}, pages 2906--2913, 2013.

\bibitem{LH05}
M.~Leordeanu and M.~Hebert.
\newblock {A spectral technique for correspondence problems using pairwise
  constraints}.
\newblock In {\em IEEE International Conference on ICCV}, pages 1482--1489,
  2005.

\bibitem{LHS09}
M.~Leordeanu, M.~Hebert, and R.~Sukthankar.
\newblock {An Integer Projected Fixed Point Method for Graph Matching and MAP
  Inference}.
\newblock In {\em Proc. NIPS}. Springer, December 2009.

\bibitem{L82}
E.~M. Luks.
\newblock {Isomorphism of Graphs of Bounded Valence Can be Tested in Polynomial
  Time}.
\newblock {\em Journal of Computer and System Sciences}, 25(1):42--65, 1982.

\bibitem{M07}
F.~M\'{e}moli.
\newblock {On the use of Gromov-Hausdorff Distances for Shape Comparison}.
\newblock In {\em Symposium on Point Based Graphics}, pages 81--90, 2007.

\bibitem{MBW13}
D.~Morozov, K.~Beketayev, and G.~Weber.
\newblock Interleaving distance between merge trees.
\newblock In {\em TopoInVis13}, 2013.
\newblock Full verstion at
  http://www.mrzv.org/publications/interleaving-distance-merge-trees/.

\bibitem{OE11}
U.~Ozertem and D.~Erdogmus.
\newblock Locally defined principal curves and surfaces.
\newblock {\em Journal of Machine Learning Research}, 12:1249--1286, 2011.

\bibitem{SPK11}
T.~Sousbie, C.~Pichon, and H.~Kawahara.
\newblock The persistent cosmic web and its filamentary structure -- {II}.
  {I}llustrations.
\newblock {\em Mon. Not. R. Astron. Soc.}, 414:384--403, 2011.

\bibitem{U98}
S.~Umeyama.
\newblock {An eigendecomposition approach to weighted graph matching problems}.
\newblock In {\em IEEE Trans. on PAMI}, volume~10, pages 695--703, 1998.

\bibitem{WW04}
B.~J. van Wyk and M.~A. van Wyk.
\newblock {A pocs-based graph matching algorithm}.
\newblock In {\em IEEE Trans. on PAMI}, volume~26, pages 1526--1530, 2004.

\bibitem{ZS08}
R.~Zass and A.~Shashua.
\newblock {Probabilistic graph and hypergraph matching}.
\newblock In {\em IEEE Conference on CVPR}, pages 1--8, June 2008.

\bibitem{ZTWFZ09}
Z.~Zeng, A.~K.~H. Tung, J.~Wang, J.~Feng, and L.~Zhou.
\newblock Comparing stars: On approximating graph edit distance.
\newblock {\em Proc. VLDB Endow.}, 2(1):25--36, Aug. 2009.

\end{thebibliography}






\end{document}